\documentclass[12pt]{amsart}
\usepackage{amsmath,amssymb,amsbsy,amsfonts,latexsym,amsopn,amstext,
                                               amsxtra,euscript,amscd}
\newfont{\teneufm}{eufm10}
\newfont{\seveneufm}{eufm7}
\newfont{\fiveeufm}{eufm5}
\newfam\eufmfam
                \textfont\eufmfam=\teneufm \scriptfont\eufmfam=\seveneufm
                \scriptscriptfont\eufmfam=\fiveeufm
\def\bbbc{{\mathchoice {\setbox0=\hbox{$\displaystyle\rm C$}\hbox{\hbox
to0pt{\kern0.4\wd0\vrule height0.9\ht0\hss}\box0}}
{\setbox0=\hbox{$\textstyle\rm C$}\hbox{\hbox
to0pt{\kern0.4\wd0\vrule height0.9\ht0\hss}\box0}}
{\setbox0=\hbox{$\scriptstyle\rm C$}\hbox{\hbox
to0pt{\kern0.4\wd0\vrule height0.9\ht0\hss}\box0}}
{\setbox0=\hbox{$\scriptscriptstyle\rm C$}\hbox{\hbox
to0pt{\kern0.4\wd0\vrule height0.9\ht0\hss}\box0}}}}
\def\bbbq{{\mathchoice {\setbox0=\hbox{$\displaystyle\rm
Q$}\hbox{\raise 0.15\ht0\hbox to0pt{\kern0.4\wd0\vrule
height0.8\ht0\hss}\box0}} {\setbox0=\hbox{$\textstyle\rm
Q$}\hbox{\raise 0.15\ht0\hbox to0pt{\kern0.4\wd0\vrule
height0.8\ht0\hss}\box0}} {\setbox0=\hbox{$\scriptstyle\rm
Q$}\hbox{\raise 0.15\ht0\hbox to0pt{\kern0.4\wd0\vrule
height0.7\ht0\hss}\box0}} {\setbox0=\hbox{$\scriptscriptstyle\rm
Q$}\hbox{\raise 0.15\ht0\hbox to0pt{\kern0.4\wd0\vrule
height0.7\ht0\hss}\box0}}}}
\def\bbbt{{\mathchoice {\setbox0=\hbox{$\displaystyle\rm
T$}\hbox{\hbox to0pt{\kern0.3\wd0\vrule height0.9\ht0\hss}\box0}}
{\setbox0=\hbox{$\textstyle\rm T$}\hbox{\hbox
to0pt{\kern0.3\wd0\vrule height0.9\ht0\hss}\box0}}
{\setbox0=\hbox{$\scriptstyle\rm T$}\hbox{\hbox
to0pt{\kern0.3\wd0\vrule height0.9\ht0\hss}\box0}}
{\setbox0=\hbox{$\scriptscriptstyle\rm T$}\hbox{\hbox
to0pt{\kern0.3\wd0\vrule height0.9\ht0\hss}\box0}}}}
\def\bbbs{{\mathchoice
{\setbox0=\hbox{$\displaystyle     \rm S$}\hbox{\raise0.5\ht0\hbox
to0pt{\kern0.35\wd0\vrule height0.45\ht0\hss}\hbox
to0pt{\kern0.55\wd0\vrule height0.5\ht0\hss}\box0}}
{\setbox0=\hbox{$\textstyle        \rm S$}\hbox{\raise0.5\ht0\hbox
to0pt{\kern0.35\wd0\vrule height0.45\ht0\hss}\hbox
to0pt{\kern0.55\wd0\vrule height0.5\ht0\hss}\box0}}
{\setbox0=\hbox{$\scriptstyle      \rm S$}\hbox{\raise0.5\ht0\hbox
to0pt{\kern0.35\wd0\vrule height0.45\ht0\hss}\raise0.05\ht0\hbox
to0pt{\kern0.5\wd0\vrule height0.45\ht0\hss}\box0}}
{\setbox0=\hbox{$\scriptscriptstyle\rm S$}\hbox{\raise0.5\ht0\hbox
to0pt{\kern0.4\wd0\vrule height0.45\ht0\hss}\raise0.05\ht0\hbox
to0pt{\kern0.55\wd0\vrule height0.45\ht0\hss}\box0}}}}
\def\bbbz{{\mathchoice {\hbox{$\sf\textstyle Z\kern-0.4em Z$}}
{\hbox{$\sf\textstyle Z\kern-0.4em Z$}} {\hbox{$\sf\scriptstyle
Z\kern-0.3em Z$}} {\hbox{$\sf\scriptscriptstyle Z\kern-0.2em
Z$}}}}

\newtheorem{theorem}{Theorem}
\newtheorem{lemma}[theorem]{Lemma}

\newtheorem{cor}[theorem]{Corollary}

\newtheorem{question}[theorem]{Open Question}

\def\squareforqed{\hbox{\rlap{$\sqcap$}$\sqcup$}}
\def\qed{\ifmmode\squareforqed\else{\unskip\nobreak\hfil
\penalty50\hskip1em\null\nobreak\hfil\squareforqed
\parfillskip=0pt\finalhyphendemerits=0\endgraf}\fi}

\def\cA{{\mathcal A}}
\def\cB{{\mathcal B}}
\def\cC{{\mathcal C}}
\def\cD{{\mathcal D}}
\def\cE{{\mathcal E}}

\def\cG{{\mathcal G}}
\def\cH{{\mathcal H}}
\def\cI{{\mathcal I}}
\def\cJ{{\mathcal J}}

\def\cL{{\mathcal L}}

\def\cO{{\mathcal O}}
\def\cP{{\mathcal P}}
\def\cQ{{\mathcal Q}}

\def\cS{{\mathcal S}}
\def\cT{{\mathcal T}}
\def\cU{{\mathcal U}}
\def\cV{{\mathcal V}}

\def\cX{{\mathcal X}}
\def\cY{{\mathcal Y}}

\def \aaf {\mathfrak a}
\def \bbf {\mathfrak b}
\def \mf {\mathfrak m}

\def \qf {\mathfrak q}
\def \tf {\mathfrak t}
\def \sf {\mathfrak s}

\def\Nm{\mbox{\rm{Nm}}\,}

\def\Oes{\cO_{e,s}}
\def\Oesx{\cO_{e,s}(x)}
\def\OAs{\cO_{\cA,s}}
\def\Oet{\cO_{e,t}}

\def\Ge{\cG_e}

\def\ZK{\Z_\K}
\def\LH{\cL_H}

\def\Im{{\mathrm{Im}}}
\def\Res{{\mathrm{Res}}}

\newcommand{\ignore}[1]{}

\def\vec#1{\mathbf{#1}}

\def \C{\mathbb{C}}
\def \F{\mathbb{F}}
\def \K{\mathbb{K}}

\def \Z{\mathbb{Z}}

\def \Q{\mathbb{Q}}

\def \Z{\mathbb{Z}}

\def\mand{\qquad\mbox{and}\qquad}

\def\Fp{\F_p}

\def\\{\cr}
\def\({\left(}
\def\){\right)}
\def\fl#1{\left\lfloor#1\right\rfloor}
\def\rf#1{\left\lceil#1\right\rceil}

\def\eps{\varepsilon}

\def\ind{\operatorname{ind}}

\begin{document}

%\title{On the Blackbox Shifted Power Identity Testing}
\title{On the Hidden Shifted Power Problem}

\author[J.~Bourgain]{Jean~Bourgain}
\address{Institute for Advanced Study,
Princeton, NJ 08540, USA}
\email{bourgain@ias.edu}

\author[M. Z. Garaev]
{Moubariz~Z.~Garaev}
\address{Centro de Ciencias Matem\'{a}ticas, Universidad Nacional Aut\'onoma de M\'{e}xico,
C.P. 58089, Morelia, Michoac\'{a}n, M\'{e}xico}
\email{garaev@matmor.unam.mx}

\author[S.~V.~Konyagin]{Sergei V.~Konyagin}
\address{Steklov Mathematical Institute,
8, Gubkin Street, Moscow, 119991, Russia}
\email{konyagin@mi.ras.ru}

\author[I.~E.~Shparlinski]{Igor E.~Shparlinski}
\address{Department of Computing, Macquarie University,
Sydney, NSW 2109, Australia}
\email{igor.shparlinski@mq.edu.au}

\begin{abstract} We consider the problem of recovering a hidden element $s$ of a finite field  $\F_q$ of $q$ elements
from queries to an oracle that for a given $x\in \F_q$ returns
$(x+s)^e$ for a given divisor $e\mid q-1$. We use some techniques
from additive combinatorics and analytic number theory that lead
to more efficient algorithms than the  naive interpolation
algorithm, for example, they use substantially fewer queries to
the oracle.
\end{abstract}

%
%\keywords{power generator, binomial exponential sums, discrepancy}
%\subjclass[2010]{11K45, 11T23, 65C10, 94A60}

\maketitle

\section{Introduction}

\subsection{Set-up and Motivation}

Let $\F_q$ be  a finite field of  $q$  elements.

For a positive integer $e\mid q-1$ and an element $s \in \F_q$ we
use  $\Oes$ an oracle that on every input $x \in \F_q$ outputs
$\Oesx = (x+s)^e$ for some ``hidden'' element $s \in \F_q$.

Here we consider the {\it Hidden Shifted Power Problem\/}:
\begin{quote}
given an oracle  $\Oes$ for some unknown $s\in \F_q$,
find $s$.
\end{quote}

We also consider the following two versions of  the
{\it Shifted Power Identity Testing\/}:
\begin{quote}
given an oracle  $\Oes$ for some unknown $s\in \F_q$ and known
$t\in \F_q$, decide whether $s = t$
provided that the call $x=-t$ is forbidden;
\end{quote}
and
\begin{quote}
given two oracles  $\Oes$ and $\Oet$ for some unknown $s,t\in \F_q$
decide whether $s = t$.
\end{quote}
Certainly these problems are special cases of the
more general problems of oracle (also sometimes called ``black-box'') polynomial interpolation
and identity testing for arbitrary  polynomials,
see~\cite{BMS} and references therein.

We note that giving the values of $(x+s)^e$ is fully equivalent
(modulo solving a discrete logarithm problem in the subgroup of $\F_q$
of order $(q-1)/e$)
to giving the values of $\chi(x+s)$ for some fixed multiplicative character
$\chi$ of $\F_q^*$, see~\cite{vDam,vDamHallgIp,RusShp},
where several classical and
quantum algorithms for this and some other similar problems are
given. The Hidden Shifted Power Problem,
under the name of {\it Hidden Root Problem\/}, has also been
re-introduced by Vercauteren~\cite{Ver}
in relation to the so-called fault attack on pairing based cryptosystems
on elliptic curves.

In the case when $\F_q$
%Igor added after subm
% a field of small characteristic,
has a subfield of an appropriate size
some
approaches to solving the  Hidden Shifted Power Problem have been
given in~\cite{Ver}. Here we concentrate on the case of  prime
fields.

For a prime $q=p\ge 3$ and $e = (p-1)/2$
the Hidden Shifted Power Problem has several other links to cryptography,
and been considered in a number of
works, see~\cite{AG,BL,Damg,HoLi} and references therein.

Furthermore,  although for application to pairing based
cryptography we usually have to solve the Hidden Shifted Power
Problem in extension fields $q = p^k$ with $k> 1$, it has been
shown by Koblitz and Menezes~\cite{KoMe} that there are elliptic
curves that lead to the case of $q=p$.

Certainly the most straightforward approach is to query  $\Oes$
on $e+1$ arbitrary  elements $x \in \F_q$ and
then interpolate the results. Using a  fast interpolation
algorithm, see~\cite{vzGG} leads to a deterministic algorithm of
complexity $e (\log q)^{O(1)}$.
For the Shifted Power Identity Testing, there is also a trivial probabilistic
algorithm that is based on querying $\Oes$ (and $\Oet$) at randomly
chosen elements $x \in \F_q$.

 Here we mainly concentrate on the
case of a prime  $q=p\ge 3$.
For the first variant of the Shifted Power Identity Testing
(that is, when) $t$ is known,
 using~\cite[Theorem~1]{BKS1} (see also~\cite{BKSc})
that gives an upper bound on
the intersection a conjugacy class of a subgroup of $\F_p^*$
with a set of Farey fractions of a given order,
we can obtain a faster algorithm of
complexity $e^{1/2} p^{o(1)}$, where $o(1)$ always, if the opposite is not
indicated, denotes a quantity that tends to zero as $p\to \infty$.

Here we obtain further improvements and in particular show that
there is an algorithm of complexity  $e^{1/4} p^{o(1)}$
for any $e\le (p-1)/2$.

The second question,  that is, when $t$ is unknown, seems to be
harder, however we also obtain an improvement of the trivial
interpolation algorithm and show that it can be solved by an
algorithm of complexity  $e^{2/3} p^{o(1)}$ for any $e\le
(p-1)/2$.
Moreover, if $e=p^{o(1)}$ then we can achieve complexity
$e^{o(1)}(\log p)^{O(1)}$.

\subsection{Our Approach}

Let $\Ge\subseteq \F_q^*$ be the multiplicative group of
order $e\mid q-1$,  that is,
$$
\Ge = \{\mu\in \F_q~:~ \mu^e = 1\}.
$$
We now define the polynomials
$$
F_{s,t}(X) = \prod_{\mu \in \Ge}\(X+s - \mu (X+t)\).
$$

Our approach is based on the idea of choosing a small ``test'' set
$\cX$, which nevertheless is guaranteed to contain at least one
non-zero of the polynomial $F_{s,t}$ for any $s\ne t$.
This is based on a  careful examination of the roots of
$F_{s,t}$ and relating it to some classical number theoretic
problems about the distribution of elements of
small subgroups of finite fields.

Clearly, if $F_{s,t}(x) = 0$ for some $x\in \F_q^*$
then
\begin{equation}
\label{eq:Cond Prelim}
\frac{x+s}{x+t} \in \Ge
\end{equation}
(provided $x+t \ne 0$).
If $t$ is known, then we can choose the ``test'' set $\cX$ in the form
\begin{equation}
\label{eq:XY}
\cX = \{y^{-1} - t~:~ y\in \cY\}
\end{equation}
for some set $\cY \subseteq \F_q^*$.
Then the condition~\eqref{eq:Cond Prelim}
means that a shift of $\cY$ is contained inside of a  coset of $\Ge$,
that is
\begin{equation}
\label{eq:Cond Fin}
\cY + r \subseteq r\Ge
\end{equation}
where $r = (s-t)^{-1}$.

So our goal is to find a ``small'' set $\cY \subseteq \F_q^*$ such
that its shifts  cannot be inside of any coset of $\Ge$ (we note
that the value of $r$ is unknown). Questions about the
distribution of cosets of multiplicative groups have been
considered in a number of works and have numerous applications,
see~\cite{KoSh} and
also~\cite{Bour,BFKS,BKPS,BKS1,BKS2,OstShp,Shp1,Shp2} for several
more recent results and applications to cryptographic and
computational number theory problems.

 Here we concentrate on
the case of prime fields, that is, when $q=p$ is prime,  where the
tools we use are most developed and have rather sharp and explicit
forms. This allows us to get a series of nontrivial estimates for
both versions of the Shifted Power Identity Testing.

The idea is to choose $\cY$ as a short interval of $h$ consecutive
integers and to define $\cX$ by~\eqref{eq:XY}. We then use a
combination of results of Cilleruelo and Garaev~\cite{CillGar}
with the classical {\it Burgess and  Weil bounds\/}
(see~\cite{IwKow}) to show that~\eqref{eq:Cond Fin} fails (for
some integer $h$ significantly smaller than $e$).

Furthermore,
for small values of $e$ (for example, for $e=p^{o(1)}$) we obtain much
stronger results and  develop
a new technique,  which is based on several tools  of commutative
algebra and additive combinatorics.
For example, we  combine an explicit version of the
Hilbert's {\it Nullstellensatz\/}, see~\cite[Theorem~1]{KiPaSo},
with a generalisation of a result of~\cite{BKS1}.

For the Hidden Shifted Power Problem we have not been able to
improve on the interpolation approach. However, assuming that
oracle calls are expensive, one can consider algorithms that
minimise the number of such calls, that is, algorithms of low {\it
oracle complexity\/}. Here we use a result of~\cite{ShkVyu} in
a combination of some new bounds of character sums that are based
on some ideas of Chang~\cite{Chang1} to
design several  algorithms that require substantially less than $e$ oracle
calls that are needed for the interpolation approach.

Here we concentrate on the case of prime $q=p$ as in the general
case several tools that exist in prime fields are unfortunately
not available.

Besides concrete results we believe the present paper also introduces a
number of new techniques to this area that can probably be used in
several other questions.

\subsection{Notation}

Throughout the paper,
the letter $p$   always denotes a prime;
$k$, $m$ and $n$ (as well as $K$, $M$ and $N$)
always denote positive integers.

Any implied constants in symbols $O$, $\ll$
and $\gg$ may occasionally depend, where obvious, on the
integer parameter $\nu$ and the real positive
parameters  $\eps$ and $\delta$, and are absolute otherwise. We recall
that the notations $U = O(V)$,  $U \ll V$ and  $V \gg U$  are
all equivalent to the statement that $|U| \le c V$ holds
with some constant $c> 0$.

For a field $\F$, sets $\cA_1,\ldots,\cA_m  \subseteq \F$ and a
rational function
$$
F(X_1, \ldots, X_m) \in\F(X_1, \ldots, X_m),
$$
we define the set
$$
F(\cA_1,\ldots,\cA_m) =\{F(a_1, \ldots, a_m) \ : \ a_1 \in \cA_1, \ldots,  a_m
\in \cA_m\}
$$
(where the poles are ignored or alternatively the function $F$ can
be defined as zero at its poles).
In particular, for an integer
$\nu$, $\cA^{(\nu)}$   denotes $\nu$-fold product sets.
However, we reserve
the notation $\nu\cA$ for the element-wise multiplication by $\nu$, that
is, $\nu \cA = \{\nu a~:~a \in \cA\}$.
We also  reserve $\cA^\nu$ for  the $\nu$-fold Cartesian product of $\cA$.

%Some of our results depend on the Extended
%Riemann Hypothesis which we
%abbreviate as ERH.

\section{Tools from Analytic Number Theory, Polynomial Algebra
and Arithmetic Combinatorics}

\subsection{Finding and Bounding the Number of Solutions of Some Congruences}
\label{sec:Congr}

We start with the bound of Cilleruelo and
Garaev~\cite[Theorem~1]{CillGar} on the number of points of
modular hyperbolas in small boxes.

 \begin{lemma}
\label{lem:CillGar} Uniformly over  integers $u$ and $v$ with
$\gcd(v,p)=1$, the congruence
$$
(x+u)(y+u) \equiv v \pmod p, \qquad   1 \le x,y \le H,
$$
has at most $H^{3/2}p^{-1/2} + H^{o(1)}$ solutions as $H\to\infty$.
\end{lemma}

We also need  an estimate from~\cite{CSZ} that follows from
a combination of a result of  Garaev and Garcia~\cite{GarGar}
(or a slightly weaker result of Ayyad, Cochrane and Zheng~\cite[Theorem~1]{ACZ})
and Lemma~\ref{lem:CillGar}.

 \begin{lemma}
\label{lem:ACZ} Uniformly over  integers $a$ and $H$ with
$\gcd(v,p)=1$, the congruence
$$
(a+x_1)(a+x_2) \equiv (a+x_3)(a+x_4)  \pmod p, \qquad   1 \le x_1,x_2,x_3,x_4 \le H,
$$
has   $H^{4}/p + O(H^{2+o(1)})$ solutions as $H\to\infty$.
\end{lemma}

The following result for $m=1$  is due to Garcia and Voloch~\cite{GV};
another proof, with different constants,  based on the method of Stepanov,  can be found
in~\cite[Lemma~3.2]{KoSh}. For any fixed $m \ge  1$ it follows instantly
from~\cite[Lemma~4.1]{ShkVyu} by taking  $s=1$, $t=e$, $k = m$
and $B = \fl{t^{1/(2k+1)}}+1$.

\begin{lemma}
\label{lem:ShkVyu} Assume that for a fixed integer $m\ge 1$ we have
$$
p \ge \(2m \fl{e^{1/(2m+1)}}+2m + 2\) e.
$$
Then for  pairwise distinct  $\mu_1,\ldots, \mu_m \in \F_p^*$
and arbitrary $\lambda_1,\ldots, \lambda_m \in \F_p^*$ the bound
$$
\# \( \cG_e \cap  \(\lambda_1 \cG_e + \mu_1\) \cap \ldots
\cap  \(\lambda_m \cG_e + \mu_m\) \)
\ll e^{(m+1)/(2m+1)}
$$
holds, where the implied constant depends on $m$.
\end{lemma}

For $x\in\F_p$, we denote by  $|x|$  the minimum of
absolute values of integers in the residue  class of $x$ modulo $p$.
We say that a set $\cD \subseteq \F_p$ is $\Delta$-spaced
if $|d_1 - d_2|\ge \Delta$ for any two distinct elements $d_1, d_2\in \cD$.

We now need a version of~\cite[Lemma~7]{Chang1}.

\begin{lemma}
\label{lem:ChangL7} Let $0 \le \beta <1/2$, $\cI = [1, p^\beta]$
and let $\cD \subseteq \F_p$ be a $p^\beta$-spaced set of
cardinality $\# \cD = p^\sigma$. Then for any $\varepsilon >0$
and sufficiently small $\beta_1, \ldots, \beta_j$ and sufficiently large $p$,
the number of solutions $w(u)$ to the congruence
$$
x+d\equiv u z_1 \ldots z_j \pmod p,
$$
with
$$
(x,d,z_1, \ldots, z_j) \in \cI\times \cD \times \cJ_1
\times \ldots \times  \cJ_j,
$$
where  $\cJ_i = [1, p^{\beta_i}]$, $i =1, \ldots, j$,
satisfies the bound
$$
\sum_{u \in \F_p} w(u)^2 \ll  p^{\beta + b + \sigma (b/(1-\beta) + 1) + \varepsilon},
$$
where
$$
b = \sum_{i=1}^j \beta_i.
$$
\end{lemma}

\begin{cor}
\label{cor:ChangL7} Let $\cS \subseteq \F_p$ be  set of
cardinality $\# \cS = p^\alpha$. Then under the conditions
of Lemma~\ref{lem:ChangL7} the number of solutions
$w(u,v)$ to the systems of  congruences
$$
x+d\equiv u z_1 \ldots z_j \pmod p\mand
x+s\equiv v z_1 \ldots z_j \pmod p
$$
with
$$
(x,d,s,z_1, \ldots, z_j) \in \cI\times \cD \times \cS \times \cJ_1
\times \ldots \times  \cJ_j,
$$
satisfies the bound
$$
\sum_{u,v \in \F_p} w(u,v)^2 \ll  p^{\alpha + \beta + b + \sigma (b/(1-\beta) + 1) + \varepsilon}.
$$
\end{cor}

\subsection{Finding Solutions to Binomial Equations}

\begin{lemma}
\label{lem:Eqgroup1} Let $G$ be a group of order $m$,
and let $d$ be relatively prime to $m$. Let $a\in G$.
Then the equation $x^d=a$ has the unique solution
$x=a^f$ where $df\equiv1\pmod m$.
\end{lemma}
This is the first part of~\cite[Theorem~7.3.1]{BaSh}.

Now we consider equaitons $x^r=a$ in groups in the case when
$r$ is a prime dividing the order of the group.
Considering the cyclic group of order $m$, we do not assume that
we are given a generating element
of the group. Instead, we assume that there is an oracle which
gives some unique label to every elements  of $G$ and also that given
$a,b \in G$ computes the product of these elements in time $(\log m)^{O(1)}$.
A natural example is the multiplicative group $\F_p^*$.
The following result is implicitly contained in~\cite[Theorem~7.3.2]{BaSh}.

\begin{lemma}
\label{lem:Eqgroup2} Let $G$ be the cyclic group of order $m$,
and let $r$ be a prime dividing $m$. Given an element $b\in G$
so that the equation $y^r=b$ has no solutions in $G$, for
any $a\in G$ there is a deterministic algorithm
to find all solutions of the equation $x^r = a$ in time
$r (\log m)^{O(1)}$.
\end{lemma}

Although the algorithm analysed in~\cite[Theorem~7.3.2]{BaSh}
is probabilistic,
is easy to see that the only place where the randomisation is used
is in finding $b$ satisfying the conditions of Lemma~\ref{lem:Eqgroup2}.

Subsequently, applying Lemma~\ref{lem:Eqgroup2} we get the following:

\begin{lemma}
\label{lem:BinEq} For a prime $p$,  a positive integer
$e\mid p-1$  and $A\in \F_p$, given $\ell$-th power nonresidues
for all prime divisors $\ell\mid e$,
there is a deterministic algorithm
to find all solutions of the equation $x^e = A$ in time
$e (\log p)^{O(1)}$.
\end{lemma}

Now we consider the solutions of the equation $x^r = A$
satisfying restrictions. Let $\ell$ be a prime divisor of $e$.
For a positive integer $\alpha$, we write $\ell^\alpha\|e$
if $\ell^\alpha\mid e$ and $\ell^{\alpha+1}\nmid e$. By $\ind x$
%\comm{How to type "not divides?"}
we denote the index of an element $x\in\F_p^*$ with respect to
a fixed primitive root $g$ modulo $p$, that is the unique integer $z\in [1,p-1]$
with $x = g^z$.

\begin{lemma}
\label{lem:BinEq2} For a prime $p$,
$A\in \F_p$, and  a prime $\ell$
with $\ell^\alpha\| p-1$, there is a deterministic algorithm
to find all solutions of the equation $x^\ell = A$
satisfying $\ell^\alpha\mid\ind x$ in time $\ell (\log p)^{O(1)}$.
\end{lemma}

\begin{proof} It is enough to apply Lemma~\ref{lem:Eqgroup1} to the group
$G=\{x\in\F_p^*~:~\ell^\alpha\mid\ind x\}$ of order $(p-1)/\ell^\alpha$
not divisible by $\ell$.
\end{proof}

\begin{lemma}
\label{lem:BinEq3} For a prime $p$,
$A\in \F_p$, for a prime divisor $\ell\mid p-1$
with $\ell^\alpha\|p-1$ and a nonnegitive integer $\beta<\alpha$,
given an $\ell^{\beta+1}$-th power nonresidue,
there is a deterministic algorithm
to find all solutions of the equation $x^\ell = A$
satisfying $\ell^{\beta}\mid\ind x$ in time $\ell (\log p)^{O(1)}$.
\end{lemma}

\begin{proof} Let $a$ be an $\ell^{\beta+1}$-th power nonresidue.
Then $\ell^\gamma\|\ind a$ for some $\gamma\le\beta$.
Hence, $\ell^{\beta}\|\ind b$ for $b=a^{l^{\beta-\gamma}}$.
Then we can apply Lemma~\ref{lem:Eqgroup2} to
$$
G=\{x\in\F_p^*~:~\ell^{\beta}\mid\ind x\}
$$
and $r=\ell$.
\end{proof}

Subsequently applying Lemmas~\ref{lem:BinEq2} and~\ref{lem:BinEq3}
we get the following.

\begin{lemma}
\label{lem:BinEq4} Let $p$ be a prime and $e\mid p-1$.
For any prime divisor $\ell\mid e$ with $\ell^{\alpha_l}\|p-1$
we take either $\gamma_\ell=\alpha_\ell$ or $\gamma_\ell<\alpha_\ell$
so that we are given
an $\ell^{\gamma_\ell+1}$-th power nonresidue. Let
$$n=\prod_{\substack{\ell\mid e\\\ell~\mathrm{prime}}}\ell^{\gamma_\ell}$$
and $A\in \F_p$. Then there is a deterministic algorithm
to find all solutions of the equation $x^e = A$
satisfying $n\mid \ind x$ in time $e (\log p)^{O(1)}$.
\end{lemma}

\begin{lemma}
\label{lem:list_cand} Assume that $p, e, n$ satisfy the conditions of
Lemma~\ref{lem:BinEq4}. Let $A_0,\ldots,A_n\in\F_p$. Then there is
a deteministic algorithm to find all solutions of the system of equations
$$
(x+j)^e = A_j,\qquad j=0,\ldots,n,
$$
in time $e (\log p)^{O(1)}n^{O(1)}$.
\end{lemma}

\begin{proof} If $A_j=0$ for some $j$, then there is nothing to prove.
We consider that $A_j\neq0$ for all $j=0,\ldots,n$. Let $x$ be any solution
of the system. By the pigeonhole principle, there are $j_1\neq j_2$
so that
$$
\ind(x+j_1)\equiv\ind(x+j_2)\pmod n,
$$ or, equivalently,
$n\mid\ind y$ for $y=(x+j_2)/(x+j_1)$. We can extract all such
$x$ satisfying the above system of equations by applying Lemma~\ref{lem:BinEq4} to the equation
$y^e=A_{j_2}/A_{j_1}$ and testing all possible values of
$$x=\frac{j_2-j_1}{y-1}-j_1.$$
To complete the proof, we simply try all pairs $(j_1,j_2)$ with $0 \le j_1<j_2\le n$.
\end{proof}

%Finally, we recall the following well-known result about
%finding solutions to
%binomial equations, which is essentially~\cite[Theorem~7.3.2]{BaSh},
%see also~\cite{AMM,vzG,Ron,Sho1}.

%\begin{lemma}
%\label{lem:BinEq} For a prime $p$,  a positive integer
%$e\mid p-1$  and $A\in \F_p$, given $\ell$-th power nonresidues
%for all prime divisors $\ell\mid e$,
%there is a deterministic algorithm
%to find all solutions of the equation $x^e = A$ in time
%$e (\log p)^{O(1)}$.
%\end{lemma}

%\begin{proof}
%Although the algorithm analysed in~\cite[Theorem~7.3.2]{BaSh}
%is probabilistic,
% is easy to see that the only place where the randomisation is used
%is in finding $\ell$-th power nonresidues
%for all prime divisors $\ell\mid e$.
%\end{proof}

\subsection{Smooth Numbers and Their Reductions Modulo $p$}

Let $x, y>0$. A positive integer $n$ is called $y$-smooth
if it is composed of prime numbers  up to $y$. The $\Psi(x,y)$
function is defined as the number of $y$-smooth positive integers
that are up to $x$.

We know the following estimate for $\Psi(x,y)$,
see~\cite[Corollary~1.3]{HT}:

\begin{lemma}
\label{lem:smoothmain2} Let $x\ge y\ge 2$ and $u=(\log x)/\log y$.
For any fixed $\delta>0$ we have
$$\Psi(x,y)=xu^{-(1+o(1))u},$$
as $y$ and $u$ tend to infinity, uniformly in the range
$y\ge(\log x)^{1+\delta}$.
\end{lemma}

\begin{cor}
\label{cor:genset} Let $0<\eps<1/2$ be fixed and $p$ be a prime.
Then the order of the subgroup of $\F_p^*$ generated by
$\{1,\ldots,\fl{p^\eps}\}$, is at least $p\eps^{-c/\eps}$
for some absolute constant $c>0$.
\end{cor}

\begin{proof} Let $x=p-1$ and $y=\fl{p^\eps}$. Also,
let $\cH$ be the subgroup of $\F_p^*$
generated by $\{1,\ldots,y\}$. If $y<(\log p)^2$ then the result
follows from the trivial estimate $\#\cH\ge1$. Assume that $y\ge(\log p)^2$.
Observe that all $y$-smooth numbers belong to $H$. Hence,
$\#\cH\ge\Psi(x,y)$, and the result follows from Lemma~\ref{lem:smoothmain2}.
\end{proof}

\subsection{Combinatorial Estimates}

We need the following result about covering an arbitrary set
$\cS \subseteq \F_p$  by  $\sqrt p/3$-spaced sets.

\begin{lemma}
\label{lem:Sets} Let $p\ge 37$, $\kappa>0$ and $\xi=\fl{\sqrt p}$.
Then any set $\cS \subseteq \F_p$ of size $\# \cS \ge 16 p^{2 \kappa}$
contains disjoint
subsets $\cD_k$, $k =1, \ldots, K$, and $\cE_\ell$, $\ell =1, \ldots, L$,
such that
\begin{itemize}
\item[{\bf (i)}] $\#\cD_k, \# \cE_\ell \ge 0.25 p^{-\kappa} \(\# \cS\)^{1/2}$,
$k =1, \ldots, K$,  $\ell =1, \ldots, L$;
\item[{\bf (ii)}] $ \cD_k$ is a $\sqrt p/3$-spaced set, $k =1, \ldots, K$;
\item[{\bf (iii)}]  $\xi \cE_\ell$ is a $\sqrt p/3$-spaced set,
$\ell =1, \ldots, L$;
\item[{\bf (iv)}] $\#\(\cS \setminus \(\cS_0 \cup \cS_1\)\) \le  2 p^{-\kappa} \# \cS$, where
$\cS_0 = \cup_{k=1}^K  \cD_k$, $\cS_1 = \cup_{\ell=1}^L  \cE_\ell$.
\end{itemize}
\end{lemma}

\begin{proof} Let $U=\sqrt p/3$. Extract from $\cS$ a maximum
(that is, not extendable any more)
collection of disjointed  $U$-spaced sets  $\cD_k$
with  $\#\cD_k\ge \(\# \cS\)^{1/2}$, $k =1, \ldots, K$ and
denote
$$
\cT = \cS \setminus \cS_0
$$
where, as before, $\cS_0 = \cup_{k=1}^K  \cD_k$.

Clearly
$$
\cT \subseteq \bigcup_{x \in \cX} (x +\cI),
$$
for $\cI = [-U, U]$
and some $U$-spaced set $\cX \subseteq \F_p$
with $\# \cX <  \(\# \cS\)^{1/2}$.
Let $\widetilde{\cE}_\ell$, $\ell =1, \ldots, L$,
be the collection of the sets
$\cT \cap (x +\cI)$, $x \in \cX$, for which
$\#\(\cT \cap (x +\cI)\) > p^{-\kappa}\(\# \cS\)^{1/2}$.

The total size of the remaining sets $\cT \cap (x +\cI)$, $x \in \cX$,  is
at most $p^{-\kappa}\(\# \cS\)^{1/2} \# \cX \le p^{-\kappa} \# \cS$.

Now we  take disjoint subsets $\cE_\ell\subseteq\widetilde{\cE}_\ell$
so that
$$\bigcup_{\ell=1}^L\cE_\ell=\bigcup_{\ell=1}^L\widetilde{\cE}_\ell.$$
Thus, for any $\ell=1,\ldots,L$ the set
$\cE_\ell$ is formed by all elements of $x\in\widetilde{\cE}_\ell$
belonging to no other sets $x\in\widetilde{\cE}_{j}$ and some elements
shared by $\cE_\ell$ and another set $\widetilde{\cE}_{j}$.

Any element $x\in \F_p$ belongs to at most two sets $\widetilde{\cE}_\ell$.
Moreover, any set $\widetilde{\cE}_{\ell}$ can have common elements with at most
two sets $\widetilde{\cE}_{j}$. If $\ell<j$ and
$\widetilde{\cE}_{\ell}$ and $\widetilde{\cE}_{j}$ have $n$ common
elements, we send $\fl{n/2}$ of them to $\cE_{\ell}$ and other
$\rf{n/2}$ elements   to
$\cE_{j}$.
We obtain a collection of disjoint sets  $ \cE_\ell$ of
size
$$\# \cE_\ell \ge \#\widetilde{\cE}_{\ell}/2-1
\ge p^{-\kappa}\(\# \cS\)^{1/2}/2 -1 \ge 0.25 p^{-\kappa}\(\# \cS\)^{1/2},
$$
for $\ell =1, \ldots, L$.

Hence, with $\cS_1 = \cup_{\ell=1}^L  \cE_\ell$, we have
$$
\#\(\cS \setminus \(\cS_0 \cup \cS_1\)\) \le  p^{-\kappa} \# \cS.
$$

Since $p\ge37$, we have
$$2U(\xi+1)<(2\sqrt p/3)(\sqrt p+1)<(\sqrt p-1)(\sqrt p+1)<p,$$
or $2U\times\xi<p-U$. Also, $\xi>U$. Therefore,
the set  $\xi \cI$ is $U$-spaced, and certainly
the set  $\xi \cE_\ell$ is also $U$-spaced for every
$\ell =1, \ldots, L$.
\end{proof}

\subsection{Bounds of Multiplicative Character Sums}

We need the following very special case of the Weil bound on
sums of
multiplicative characters (see~\cite[Theorem~11.23]{IwKow}).

\begin{lemma}
\label{lem:Weil1} For an arbitrary  integer $h$ with  $1 \le h <
p$, a positive integer $f$ and  a nonprincipal  multiplicative
character $\chi$ of $\F_p^*$, the bound
$$
\sum_{x = 1}^{p} \chi\(x^f+h\)  = O\(f p^{1/2}\)
$$
holds.
\end{lemma}

Also, we need an estimate for character sums including
both multiplicative and additive
characters (see~\cite[Chapter~6, Theorem~3]{Li}
or~\cite[Appendix~5, Example~12]{Weil}).

\begin{lemma}
\label{lem:Weil2} Let $\chi_1,\ldots,\chi_r$ be characters modulo
$p$, and at least one of them is nonprincipal and let $f(X) \in
\F_p[X]$ be an arbitrary polynomial of degree $d$. Then for any
distinct $a_1,\ldots,a_r\in\F_p$ we have
$$\left|\sum_{x\in\F_p}\chi_1(x+a_1)\ldots\chi_r(x+a_r)
\exp\(2 \pi i f(x)/p\)\right|\le (r+d) p^{1/2}.
$$
\end{lemma}

The standard reduction of incomplete sums to
complete ones (see~\cite[Section~12.2]{IwKow})
together with the bound of Lemma~\ref{lem:Weil2}
lead to the following estimate:

\begin{lemma}
\label{lem:Weil3} For an arbitrary  integer $h$ with  $1 \le h \le p$,
distinct elements $s,t \in \F_p$ and
 a nonprincipal  multiplicative character $\chi$ of $\F_p^*$, the bound
$$
\sum_{\substack{x = 1 \\ x \ne t}}^{h} \chi\(\frac{x+s}{x+t}\)  = O\( p^{1/2} \log p\)
$$
holds.
\end{lemma}

The following result  is a combination of  the
P{\'o}lya-Vinogradov  (for $\nu =1$) and Burgess
(for $\nu\ge2$)  bounds,
see~\cite[Theorems~12.5 and 12.6]{IwKow}.

\begin{lemma}
\label{lem:PVB} For an arbitrary  integer $h$ with  $1 \le h < p$, and
 a nonprincipal  multiplicative character $\chi$ of $\F_p^*$, the bound
$$
\left| \sum_{y = 1}^{h} \chi(y)\right|
\le h^{1 -1/\nu} p^{(\nu+1)/4\nu^2 + o(1)}
$$
holds with an arbitrary
positive integer $\nu$.
\end{lemma}
%
%Furthermore, under the  ERH a much stronger (essentially optimal)
%bound is known, see, for example,~\cite[Bound~(13.2)]{Mont},
%one can also derive it from~\cite[Theorem~2]{GrSo}.
%
%
%\begin{lemma}
%\label{lem:Char-ERH} Under the ERH, for an arbitrary  integer $h$ with  $1 \le h < p$, and
% a nonprincipal  multiplicative character $\chi$ of $\F_p^*$, the bound
%$$
%\left| \sum_{y = 1}^{h} \chi(y)\right|
%\le h^{1/2} p^{o(1)}
%$$
%holds.
%\end{lemma}

We use $\overline \chi$ to denote the complex conjugate character
to $\chi$.
The following estimate is a generalisation of~\cite[Theorem~8]{Chang1}.

\begin{lemma}
\label{lem:ChangT8} Let $0 \le \beta <1/2$, $\cI = [0, p^\beta]$.
Let $\cD \subseteq \F_p$ be a $p^\beta$-spaced set of
cardinality $\# \cD = p^\sigma$ and let $\cS \subseteq \F_p$ be  set of
cardinality $\# \cS = p^\alpha$.
For any $\delta > 0$ there is some $\eta > 0$ such that if
$$
2\beta + \alpha +  \sigma\frac{1-2\beta}{1-\beta} > 1+ \delta
$$
then for any nontrivial multiplicative character $\chi$ of $\F_p^*$ we have
$$
\sum_{d\in \cD}   \sum_{s\in \cS}  \left| \sum_{x \in \cI} \chi(x+d)\overline \chi(x+s)\right| < p^{\alpha + \beta + \sigma - \eta}.
$$
\end{lemma}

\begin{proof} We can assume that $\delta$ is sufficiently small.
Take a sufficiently large $j$ and set
$$
\beta_0 = \frac{\delta}{6} \mand \gamma = \frac{\beta-\beta_0}{j+1}
$$
and consider the intervals $\cJ=[1,p^\gamma]$ and  $\cJ_0=[1,p^{\beta_0}]$.
It is easy to see that
\begin{equation*}
\begin{split}
\sum_{d\in \cD}   \sum_{s\in \cS}  &\left|\sum_{x \in \cI}  \chi(x+d)\overline \chi(x+s)\right|\\
&\ll  p^{-\beta_0 - j\gamma} \sum_{d\in \cD}   \sum_{s\in \cS}
 \sum_{z_1, \ldots, z_j\in \cJ} \\
 & \qquad \qquad \sum_{x \in \cI}
  \left| \sum_{t \in \cJ_0}
  \chi\(\frac{x+d}{z_1 \ldots z_j} + t\)\overline \chi\(\frac{x+s}{z_1 \ldots z_j} + t\)\right|\\
  & \qquad \qquad \qquad  \qquad \qquad \qquad
   \qquad \qquad \qquad  +   p^{\alpha +  \sigma+\beta_0 + j\gamma}.
\end{split}
\end{equation*}
Now, invoking Corollary~\ref{cor:ChangL7} and
 using the same argument as in the proof of~\cite[Theorem~8]{Chang1}
we obtain the desired result.
\end{proof}

\begin{lemma}
\label{lem:TripleSums}
Assume that $\alpha > 0$, $\delta > 0$ and
$0 \le \beta <1/2 - \delta$
satisfy
$$
2\beta + \alpha \frac{3-4\beta}{2-2\beta} > 1+ \delta.
$$
Let $\cS \subseteq \F_p$ be of cardinality $\# \cS = p^\alpha$ and let
$\cI = [0, p^\beta]$.
We denote
$$
\xi=\fl{\sqrt p}
$$
and define $\zeta$ by the conditions
$$
\ \zeta\xi \equiv 1 \pmod p \mand  1 \le \zeta < p.
$$
There is a partition
$$
\cS = \cT_0 \cup \cT_1  \mand \cT_0 \cap \cT_1=\emptyset
$$
such that for any nontrivial multiplicative character $\chi$ of $\F_p^*$ we have
$$
 \sum_{s_1,s_2\in \cT_\nu}
 \left| \sum_{x\in \cI} \chi(\zeta^\nu x+s_1)\overline \chi(\zeta^\nu x+s_2)\right|
 \ll    \#\cI (\#S)^2  p^{-\eta}, \qquad \nu =0,1,
$$
for some $\eta > 0$ that depends only on $\delta$.
\end{lemma}

\begin{proof} We consider that $p$ is so large that $p\ge 37$ and $p^\delta\ge3$.
Then $p^\beta<\sqrt p/3$. We take  $\kappa = \delta/10$
and define set $\cS_0$ and $\cS_1$ as in Lemma~\ref{lem:Sets}.
We now put
$\cT_1 = \cS_1$ and then define $\cT_0 = \cS \setminus \cT_1$.

Then in the notation of Lemma~\ref{lem:Sets} we have
\begin{equation*}
\begin{split}
 \sum_{s_1,s_2\in \cT_0} &
 \left| \sum_{x\in \cI} \chi(x+s_1)\overline \chi(x+s_2)\right|\\
 & =  \sum_{s_1,s_2\in \cS_0}
 \left| \sum_{x\in \cI} \chi(x+d)\overline \chi(x+s)\right| + O\( (\#S)^2 \#\cI
 p^{-\kappa} \)\\
 & = \sum_{k=1}^K
 \sum_{d \in \cD_k} \sum_{s \in \cS_0}
 \left| \sum_{x\in \cI} \chi(x+d)\overline \chi(x+s)\right| +  O\( (\#S)^2 \#\cI
 p^{-\kappa} \)\\
 & \le  \sum_{k=1}^K
 \sum_{d \in \cD_k} \sum_{s \in \cS}
 \left| \sum_{x\in \cI} \chi(x+d)\overline \chi(x+s)\right|+   O\( (\#S)^2 \#\cI
 p^{-\kappa} \).
\end{split}
\end{equation*}
Since $\cD_k$ is a $p^\beta$-spaced of cardinality $\# \cD_k \ge p^\sigma$
with $\sigma = \alpha/2 -\kappa$, $k=1, \ldots, K$, we see that the conditions
of Lemma~\ref{lem:ChangT8} are satisfied, which implies
the desired bound for the set $\cS_0$.

For the set $\cT_1$ we write
\begin{equation*}
\begin{split}
\sum_{s_1,s_2\in \cT_1} &
 \left| \sum_{x\in \cI} \chi(\zeta x+s_1)\overline \chi(\zeta x+s_2)\right|\\
 & = \sum_{s_1,s_2\in \cS_1}
 \left| \sum_{x\in \cI} \chi(x+\xi s_1)\overline \chi(x+\xi s_2)\right|
\end{split}
\end{equation*}
and then proceed as before, applying  Lemma~\ref{lem:ChangT8}
with  $\cS$ replaced by the set $\xi \cS$.
\end{proof}

%Combining Lemma~\ref{lem:PVB} (for $h \ge 0.5 p^{1/2}$)
%and  Lemma~\ref{lem:TripleSum} (for $h \ge 0.5 p^{1/2}$),
%we obtain:
%
%\begin{cor}
%\label{cor:TripleSum} For any fixed $\varepsilon > 0$ there is some $\eta > 0$ so that
%for an arbitrary  integer $h$ with  $1 \le h < p$, a
%set $\cS \subseteq \F_p$ such that
%$$
%h^2 \#S  \ge p^{1 +\varepsilon}
%$$
%and a nonprincipal  multiplicative character $\chi$ of $\F_p^*$, the bound
%$$
% \sum_{s_1,s_2\in \cS} \left| \sum_{x = 1}^{h} \chi(x+s_1)\overline \chi(x+s_2)\right|
%\ll  h  (\#S)^2  p^{-\eta}
%$$
%holds.
%\end{cor}

\subsection{Quantitative Result on Polynomial Ideals}

We recall the following quantitative version
of the B{\'e}zout theorem, that follows from a
result of  Krick,  Pardo, and   Sombra~\cite[Theorem~1]{KiPaSo}
(that improves a series of previous estimates).

We recall that the logarithmic height of a nonzero polynomial $P \in
\Z[Z_1, \ldots, Z_n]$ is defined as the maximum
logarithm of the largest (by absolute value) coefficient of $P$.

\begin{lemma}
\label{lem:Bezout} Let
$P_1, \ldots, P_N \in \Z[Z_1, \ldots, Z_n]$
be $N\ge 1$ polynomials in $n$ variables
without common zero in $\C^n$
of degree at most $D\ge 3$ and of logarithmic height
at most $H$. Then there is a positive integer $b$
with
$$
\log b \le c(n) D^{n}\(H + \log N + D\)
$$
and polynomials $R_1, \ldots, R_N\in \Z[Z_1, \ldots, Z_n]$ such that
$$
P_1R_1+ \ldots + P_NR_N = b,
$$
where $c(n)$ depends only on $n$.
\end{lemma}

Using the classical argument of Hilbert we obtain the
following version of the Nullstellensatz (see~\cite{BBK}
for several similar results and further references).

\begin{lemma}
\label{lem:Hilb} Let $P_1, \ldots, P_N, f \in \Z[Z_1, \ldots,
Z_n]$ be $N+1\ge 2$ polynomials in $n$ variables of degree at most
$D\ge 3$ and of logarithmic height at most $H$ such that $f$
vanishes on the variety
$$
P_1(Z_1, \ldots, Z_n) = \ldots =P_N(Z_1, \ldots, Z_n) = 0.
$$
There are   positive integers $b$ and $r$
with
$$
\log b \le C(n) D^{n + 1}\(H + \log N + D\)
$$
and polynomials $Q_1, \ldots, Q_N\in \Z[Z_1, \ldots, Z_n]$ such that
$$
P_1Q_1+ \ldots + P_NQ_N = bf^r,
$$
where $C(n)$ depends only on $n$.
\end{lemma}

\begin{proof}
We consider $N+1$ polynomials
$$
\widetilde P_0=1 - T f \mand
\widetilde P_j=  T P_j, \ j =1, \ldots, N,
$$
in $\Z[Z_1, \ldots, Z_n,T]$. By the assumption on $f$, they
have no common zero. Hence,  by Lemma~\ref{lem:Bezout}
we get
$$
(1 - T f)Q_0 + TP_1Q_1 + \ldots + TP_NQ_N = b
$$
for some polynomials $Q_0,Q_1, \ldots, Q_N\in \Z[Z_1, \ldots, Z_n]$
and a positive integer $b$ satisfying the desired inequality.
Replacing $T$ by $1/f$ and clearing the denominators we obtain the
desired relation.
\end{proof}

Finally, we need a slightly more general form of a result of
Chang~\cite{Chang0}. In fact, this is exactly the statement that
is established in the proof of~\cite[Lemma~2.14]{Chang0},
see~\cite[Equation~(2.15)]{Chang0}.

\begin{lemma}
\label{lem:SmallZero} Let $P_1, \ldots, P_N, P \in \Z[Z_1, \ldots,
Z_n]$ be $N+1 \ge 2$ polynomials in $n$ variables of degree at
most $D$ and of logarithmic height at most $H\ge1$. If  the
zero-set
$$
P_1(Z_1, \ldots, Z_n) = \ldots =P_N(Z_1, \ldots, Z_n) = 0 \quad
\text{and}\quad
 P(Z_1, \ldots, Z_n) \ne  0
$$
is not empty then it has a point $(\beta_1, \ldots, \beta_n)$
in an extension $\K$ of $\Q$ of degree $[\K:\Q]\le C_1(D,N,n)$
such that the  minimal polynomials are  of  logarithmic
height at most $C_2(D,N,n) H$,
where $C_1(D,N,n)$ and $C_2(D,N,n)$ depend only on $D$, $N$ and $n$.
\end{lemma}

Finally, we recall the following well-known result, see,
for example,~\cite[Theorem~6.32]{vzGG}.

\begin{lemma}
\label{lem:HeighDiv} Let $P,Q \in \Z[Z]$ be
two univariate non-zero polynomials with $Q\mid P$.
If $P$ is of logarithmic height at most $H\ge1$
then $Q$ is of  logarithmic height at most $H + O(1)$,
where the implied constant depends only on $\deg P$.
\end{lemma}

\subsection{Product Sets in Number Fields}

Let $\K$ be a finite extension of $\Q$ and
let $\ZK$ be the ring of integers in $\K$.
We denote by $\cH(\gamma)$ the logarithmic height of $\gamma\in \K$.
We recall that the logarithmic height of an algebraic number
$\alpha$ is defined as the logarithmic height
of its minimal polynomial.

For an integral ideal $\aaf$ of $\ZK$  we denote by $\Nm(\aaf)$
the norm of $\aaf$, that is, the cardinality of the residue ring
$\ZK/\aaf \ZK$. We also use $\Nm(\alpha)$ to denote the norm of
$\alpha \in \ZK$. In particular $\Nm(\alpha) =  \Nm((\alpha))$
where $(\alpha)$ denotes the principal ideal generated by
$\alpha$.

First we recall the following well-known bound,
which follows immediately from~\cite[Lemma~4.2]{Nar}
and  the classical bound on the
divisor function.

\begin{lemma}
\label{lem:IdealNorm} Let $\K$ be a finite extension of $\Q$ of degree
$d = [\K:\Q]$. For any integer $N\ge 3$,  in $\K$ there are at most
$\exp\(O(\log N/\log \log N)\)$ integral ideals  of norm $N$,
where the implied constant depends on $d$.
\end{lemma}

We also need a bound of Chang~\cite[Proposition~2.5]{Chang0} on the divisor function in
algebraic number fields.

\begin{lemma}
\label{lem:Div ANF} Let $\K$ be a finite extension of $\Q$ of degree
$d = [\K:\Q]$. For any algebraic integer $\gamma\in \Z_K$ of logarithmic
height at most $H\ge 2$, the number of   pairs $(\gamma_1, \gamma_2)$
of  algebraic integers $\gamma_1,\gamma_2\in \Z_K$ of
logarithmic  height at most $H$
with $\gamma=\gamma_1\gamma_2$ is at most $\exp\(O(H/\log H)\)$,
where the implied constant depends on $d$.
\end{lemma}

We  now derive the following  generalisation
of~\cite[Lemma~2]{BKS1}.

\begin{lemma}
\label{lem:ProdSet} Let $\K$ be a finite extension of $\Q$ of
degree $d = [\K:\Q]$. Let $\cA, \cB \subseteq \K$ be finite sets
with elements of  logarithmic height at most $H$.
Then we have
$$
\#(\cA\cB) > \exp\(-c(d) \frac{H}{\sqrt{\log H}}\) \#\cA \# \cB,
$$
where $c(d)$ depends only on $d$.
\end{lemma}

\begin{proof}
We fix some maps $\aaf(\gamma)$
and $\bbf(\gamma)$ (not necessary uniquely defined) that for
an algebraic number $\gamma\in \K$ produce relatively prime
 ideals $\aaf(\gamma), \bbf(\gamma) \in \ZK$  of norm
$\exp\(O\(\cH(\gamma)\)\)$ with
$$
\gamma \bbf(\gamma) = \aaf(\gamma).
$$

We also use $\LH$ to denote the set of elements of $\K$
of logarithmic height at most $H$.

Clearly if the ideals $\aaf(\gamma)= \aaf$ and $\bbf(\gamma)=\bbf$
are fixed then $\gamma$ is defined up to a multiplication
by a unit. Thus as in the proof of~\cite[Proposition~2.5]{Chang0}
we see that for any integral ideals $\aaf$ and $\bbf$ we have
\begin{equation}
\label{eq:ideals}
\#\left \{\gamma \in \cL_H~:~ \aaf(\gamma)= \aaf, \ \bbf(\gamma)=\bbf \right\}
= \exp \(O\(H/\log H\)\).
\end{equation}

Denote
$$
M_1=\exp \(c_1(d)\frac{H}{\sqrt{\log H}}\) \mand M_2=\exp\(c_2(d)\frac{H}{\log H} \)
$$
for certain constants  $c_1(d),c_2(d)>0$ that depend only on $d$.

We claim that for an appropriate choice of
$c_1(d)$ and $c_2(d)$ there is a subset $\cA_0 \subseteq \LH$
of cardinality
\begin{equation}
\label{eq:Card A0}
   \#\cA_0> M_2^{-2 H/\log M_1} \# \cA=\exp
\(-2\frac{H}{\sqrt {\log H}}\) \# \cA
\end{equation}
 and two integral ideals $\sf_1$ and $\sf_2$
such that  $\sf_1 \cA_0 \subseteq \sf_2\cA$ and for any integral ideal $\mf$
with  $\Nm(\mf)>M_1$, we have
\begin{equation}
\label{eq:Div A0}
\# \left\{\gamma \in \cA_0~:~  \mf \mid \aaf(\gamma)\quad\text{or}
\quad \mf \mid \bbf(\gamma) \right\}< \frac 2{M_2} \# \cA.
\end{equation}

The construction is straightforward.
For a real positive $R$ we denote
$$
\cE_R=\left\{\gamma \in \K~:~   \Nm\( \aaf(\gamma)\bbf(\gamma)\) \leq  R\right\}.
$$
Hence $\LH \subseteq \cE_{R}$ where
\begin{equation}
\label{eq:H and R}
R = \exp\(O\(H\)\).
\end{equation}

If $\cA_0 = \cA$ does not satisfy~\eqref{eq:Div A0},
there is an integral ideal $\mf_1 \in\Z$
with $\Nm(\mf_1)> M_1$ and a subset
$\cA_1\subseteq \cE_{R/\Nm(\mf_1)}
\subseteq  \cE_{R/M_1}$ of cardinality $\# \cA_1\geq M_2^{-1} \# \cA$
and such that
\begin{itemize}
\item  either
$$\mf_1 \cA_1 \subseteq \cA$$
\item or
$$  \cA_1 \subseteq \mf_1\cA.
$$
\end{itemize}

Repeat with $\cA$ replaced by $\cA_1$ until,
after performing $k$ steps, we obtain a subset
$\cA_k\subseteq \cE_{RM_1^{-k}}$ such that
 $\sf_1 \cA_k \subseteq \sf_2\cA$  for some
two integral ideals $\sf_1$ and $\sf_2$ and
such that~\eqref{eq:Div A0} holds with $\cA_k$ instead
of $\cA_0$. Assuming that $\cA_k$ is the first set
with this property, we derive
$$
\# \cA_k \ge \frac{1}{M_2} \# \cA_{k-1}\ge
\ldots \ge  \frac{1}{M_2^{k}} \# \cA.
$$
Since we obviously have $R\geq M_1^k$, we see from~\eqref{eq:H and R}
that $k \ll H/\log M_1$
which implies~\eqref{eq:Card A0} provided that
$$
c_1(d) = \frac{1}{2} c_2(d).
$$

We now  use a similar argument to choose  a
subset $\cB_0 \subseteq \LH$
of cardinality
\begin{equation}
\label{eq:Card B0}
\#\cB_0> M_2^{-2H/\log M_1} \# \cB>\exp
\(-2\frac{H}{\sqrt {\log H}}\) \# \cB
\end{equation}
and two integral ideals $\tf_1$ and $\tf_2$
such that  $\tf_1 \cB_0 \subseteq \tf_2\cB$ and for any integral ideal $\mf$
with  $\Nm(\mf)>M_1$, we have
$$
\# \left\{ \gamma \in \cB_0~:~ \mf \mid \aaf(\gamma)\quad\text{or}
\quad \mf \mid \bbf(\gamma) \right\}< \frac 2{M_2}  \# \cB.
$$

We now establish a lower bound on  $\# \(\cA_0   \cB_0\)$.

Given $\gamma \in \LH$,  denote
\begin{equation*}
\begin{split}
\cA_0(\gamma)=\{\vartheta \in \cA_0~:~
\Nm\(\gcd(\aaf(\vartheta),\bbf(\gamma))\)\leq M_1&,\\
\Nm\(\gcd(\bbf(\vartheta), \aaf(\gamma))\)& \leq M_1\}.
\end{split}
\end{equation*}

We now recall the well known bound
on the divisor function
\begin{equation}
\label{eq:tau}
\tau(m)  \le \exp\((\log 2+o(1))
\frac{\log m}{\log \log m} \)
\end{equation}
%% see, for example,~\cite[Theorem~317]{HardyWright}.
which is also a special case of Lemma~\ref{lem:Div ANF}.

As in the proof of~\cite[Lemma~2]{BKS1}, we note
that the bounds~\eqref{eq:Div A0}, \eqref{eq:tau} and  Lemma~\ref{lem:IdealNorm}
imply that, for a sufficiently large $H$,
$$
\#\(\cA_0\backslash \cA_0(\gamma)\) \le \frac{2}{M_2} \#\cA_0 \exp \(O\(\frac{H}{\log H}\)\)
 <\frac 12 \#\cA_0
$$
for an appropriate choice of $c_2(d)$ in the definition
of $M_2$.

 Defining $\cB_0(\gamma)$ in a similar way, we conclude that
\begin{equation}
\label{eq: A0B_0 large}
\#  \cA_0(\gamma)\ >\frac 12 \# \cA_0\mand
\#  \cB_0(\gamma)\ >\frac 12 \# \cB_0
\end{equation}
for every $\gamma \in \LH$.

We have
$$
  \# \(\cA  \cB\)  \ge  \# \(\cA_0   \cB_0\)
\ge \# \(\bigcup_{\vartheta \in \cA_0}
\{\vartheta \rho~:~ \rho\in \cB_0(\vartheta)\}\).
$$

Using~\eqref{eq: A0B_0 large} we conclude that
\begin{equation}
\label{eq:prelim bound}
  \# \(\cA  \cB\)  \ge   \frac{1}{2L} \# \cA_0  \# \cB_0,
\end{equation}
where
$$
L=\max_{\gamma\in \LH}\# \left\{(\vartheta,\rho) ~:~
\vartheta\in \cA_0, \
\rho \in \cB_0(\vartheta), \  \vartheta \rho=\gamma\right\}.
$$
It remains to bound $L$.

Since
$$
\aaf(\vartheta) \aaf(\rho)\bbf(\gamma)  =
\bbf(\vartheta) \bbf(\rho) \aaf(\gamma),
$$
it follows from the definition of $\cB_0(\alpha)$
that $\aaf(\vartheta)  = \qf \mf$ for some integral ideal
$(\qf)$ dividing $\aaf(\gamma)$ and integral ideal $\mf$ with
$\Nm(\mf) \le M_1$.
We recall that there are $O(M_1)$
integral ideals of norm at most $M_1$, see~\cite[Proposition~7.10]{Nar}.
Hence
by~\eqref{eq:tau} and Lemma~\ref{lem:IdealNorm}, there are only at
most $M_1\exp \(O\(H/\log H\)\)$ possible values that can be taken
by the ideal $\aaf(\vartheta)$.

Similarly, estimates also hold for the number of possible
values that can be taken  by $\aaf(\rho)$,
$\bbf(\vartheta)$ and $\bbf(\rho)$.

We now recall that all elements $\vartheta\in \cA_0$ satisfy
$\sf_1 \vartheta = \eta \sf_2$ with $\vartheta \in \cL_H$
and fixed integral ideals $\sf_1$ and $\sf_2$.  We also have a
similar property for all elements
$\rho \in \cB_0(\vartheta) \subseteq \cB_0$.
Therefore, using~\eqref{eq:ideals}, we derive
\begin{equation}
\label{eq:L bound}
L\le M_1^{4}   \exp \(O\(\frac{H}{\log H}\)\)
 \le \exp \(5c_1(d) \frac{H}{\sqrt {\log H}}\),
\end{equation}
provided that $H$ is large enough.
Substituting~\eqref{eq:L bound}
in~\eqref{eq:prelim bound},
and using~\eqref{eq:Card A0} and~\eqref{eq:Card B0}, we conclude the proof.
\end{proof}

We also have a full analogue of~\cite[Corollary~3]{BKS1}

\begin{cor}
\label{cor:ProdSet} Let $\K$ be a finite extension of $\Q$ of
degree $d = [\K:\Q]$. Let $\cC \subseteq \K$ be a finite set with
elements of  logarithmic height at most $H\ge2$.
Then we have
$$
\#(\cC^{(\nu)}) > \exp\(-c(d,\nu) \frac{H}{\sqrt{\log H}}\) (\#\cC)^\nu,
$$
where $c(d,\nu)$ depends only on $d$ and $\nu$.
\end{cor}

\subsection{Resultant bound}

Let  $\nu\ge 2$ be an integer. For integers  $n,m$ with $2\le n,m\le \nu$,
we define the $(n+m-2)\times (m-1)$ matrix   $A(\nu; n, m)$  as follows:
\begin{equation*}
\(
  \begin{array}{cccccccc}
    \nu-n+1 & \nu-n+2 & \ldots & \nu & 0 & 0 & \ldots & 0 \\
    0 & \nu-n+1 & \ldots  & \nu-1 & \nu & 0 & \ldots & 0  \\
    \ldots & \ldots & \ldots  & \ldots & \ldots & \ldots & \ldots & \ldots \\
    0 & \ldots & 0 & \nu-n+1 & \nu-n+2 & \ldots  & \ldots &\nu\\
  \end{array}
\)
\end{equation*}
Note that each row of $A(\nu; n, m)$ contains $m-2$ zeros.

\begin{lemma}
\label{lem:DeterMagic} Let $2 \le n,  m\le \nu$ be integers.
If in the  $(n+m-2)\times (n+m-2)$ matrix
$$
X(\nu; n, m)=\(
  \begin{array}{c}
    A(\nu; n, m)\\
    A(\nu; m, n) \\
  \end{array}
\)
$$
we mark $n+m-2$ nonzero elements such that  each row and each column contains exactly one marked element then the sum of the marked elements is always equal to
$$
\sigma = (\nu-n+1)(m-1)+\nu(n-1).
$$
\end{lemma}

\begin{proof}
Let
$$
X(\nu; n, m)=(x_{i,j})_{1\le i,j\le n+m-2}
$$
where $i$ indicates the row. Since the sum of the diagonal elements of
$X(\nu; n, m)$ is equal to $(\nu-n+1)(m-1)+\nu(n-1)$, it suffices to prove
that the sum of the marked elements does not depend on the choice of marking.
To see this, we transform the matrix $X(\nu; n, m)$ into  a matrix
$$
Y(\nu; n, m)=(y_{i,j})_{1\le i,j\le n+m-2}
$$
as follows
\begin{itemize}
\item  If $x_{i,j}=0$, then we put $y_{i,j} =0$
\item If $x_{i,j}\ne 0$, then we put
$$y_{i,j} = \left\{
\begin{array}{ll}
x_{i,j} + n+2i-\nu-1, & \text{ for } 1\le i\le m-1,\\
x_{i,j} + 2i-\nu, & \text{ for }   m\le i\le n+m-2.
\end{array}
\right.
$$
\end{itemize}

Since the marked elements occur in each row exactly once, from this transformation of $X(\nu; n, m)$ into $Y(\nu; n, m)$ the sum of the  elements at the
marked positions changes only by
\begin{equation}
\label{eq:Count1}
\begin{split}
\sigma_1 =&\sum_{i =1}^{m-1} ( n+2i-\nu-1)  +\sum_{i=m}^{n+m-2}(2i-\nu)\\
=& (n-1)(m-1) - \nu(n+m-2) + 2 \sum_{i=1}^{n+m-2} i \\
=&  (n-1)(m-1) - \nu(n+m-2) + (n+m-1)(n+m-2)\end{split}
\end{equation}
and in particular does not depend on the choice of the marking. Therefore, it suffices to show that the corresponding marked elements of $Y(\nu; n, m)$ does not depend on the choice of marking.  But this follows from the observation that when $x_{ij}\not=
0$, we have that
$$
y_{i,j}=i+j.
$$
Hence, the sum of the corresponding marked elements of $Y(\nu; n, m)$ is equal to
$$
\sigma_2 = 2(1+\ldots +(n+m-2)) = (n+m-1)(n+m-2)
$$
and does not depend on the choice of marking.
Since $\sigma_2 - \sigma_1 = \sigma$, the result now follows.
\end{proof}

In particular, since $n,m\le \nu$, the sum  $\sigma$ of the marked elements in
Lemma~\ref{lem:DeterMagic}  is monotonically increasing function of $m$.
So replacing $m$ with $\nu$ we derive
$$
\sigma \le (\nu-n+1)(\nu-1)+\nu(n-1) = \nu(\nu-1) +n -1 \le \nu^2-1.
$$

\begin{cor}
\label{cor:DeterMagic} Let $M\ge1$ and
let  $2 \le n, m\le \nu$ be fixed integers. Let $P_1(Z)$ and $P_2(Z)$ be polynomials
 $$
P_1(Z)=\sum_{i=0}^{n-1}a_{i}Z^{i} \mand P_2(Z)= \sum_{i=0}^{m-1}b_{i}Z^{i}
 $$
such that
$$
a_{n-1}, b_{m-1}\not=0 \mand  |a_{i}|, |b_{i}| <M^{\nu -i}, \ i =0, \ldots, \nu-1.
$$
Then
$$
|\Res(P_1, P_2)|\ll M^{\nu^2-1}.
$$
\end{cor}

\begin{proof}
We recall that
$$
\Res(P_1, P_2)=\det\left(
  \begin{array}{c}
    A\\
    B\\
  \end{array}
\right)
$$
where
\begin{equation*}
A=\left(
  \begin{array}{cccccccc}    a_{n-1} &  \ldots &  a_1 &a_{0} & 0 & 0 & \ldots & 0 \\
    0 & a_{n-1}&  \ldots &  a_1 &a_{0} & 0 & \ldots & 0  \\
    \ldots & \ldots & \ldots  & \ldots & \ldots & \ldots & \ldots & \ldots \\
    0 & \ldots & 0 & a_{n-1} &  \ldots & \ldots &  a_1 &a_{0}\\
  \end{array}
\right),
\end{equation*}
and
\begin{equation*}
B=\left(
  \begin{array}{cccccccc}   b_{m-1} & \ldots & b_1  &  b_0  & 0 & 0 & \ldots & 0 \\
    0 & b_{m-1} & \ldots & b_1  &  b_0 & 0 & \ldots & 0  \\
    \ldots & \ldots & \ldots  & \ldots & \ldots & \ldots & \ldots & \ldots \\
    0 & \ldots & 0 & b_{m-1} & \ldots & \ldots & b_1  &  b_0\\
  \end{array}
\right).
\end{equation*}
The result now follows from the representation of the determinant by sums of products of its elements and Lemma~\ref{lem:DeterMagic}.
\end{proof}

\subsection{Product Sets in $\F_p$}

We believe the results of this section can be of independent interest
and have several other applications. For example, the following result
in the case $\nu=4$
solves an open problem from~\cite{CillGar}.

\begin{lemma}
\label{lem:GProdSet Fp lin 1} Let $\nu \ge 1$ be a fixed integer, $\lambda\not\equiv 0\pmod p$. Assume
that for some sufficiently large positive integer $h$ and prime
$p$ we have
$$
h<p^{1/\max\{\nu^2-1,1\}}.
$$
Then for any  $s \in \F_p$ for the number $J_{\nu}(\lambda; h)$ of solutions of the congruence
$$
(x_1+s)\ldots(x_{\nu}+s)\equiv \lambda \pmod p,\quad 1\le x_1,\ldots,x_{\nu}\le h,
$$
we have the bound $$J_{\nu}(\lambda;h)<\exp\(c(\nu)\frac{\log h}{\log\log h}\),$$
 where  $c(\nu)$ depends only on $\nu$.
\end{lemma}

\begin{proof}
We note that for $\nu =1$ the result it trivial and we   prove it for
$\nu \ge 2$  by induction on $\nu$.
%It suffices to prove that
%$$
%J_{\nu}(\lambda; h)\ll_{\nu} \max\Bigl\{J_{\nu-1}(\lambda'; h),\, \exp\(c(\nu)\frac{\log h}{\log\log h}\)\Bigr\},
%$$
%for some $\lambda_*\not\equiv 0\pmod p.$ Then the result will follow by induction.
%
%
%
%In what follows we shall assume that for any $\lambda_*\not\equiv 0\pmod p$ we have
%\begin{equation}
%\label{eq:InductionForJ}
%J_{\nu}(\lambda; h)> C\,J_{\nu-1}(\lambda_*; h),
%\end{equation}
%where $C=C(\nu)$ is a large constant.

Let $\varepsilon<1$ be a sufficiently small positive number, to be chosen later.
We split the interval $[1,h]$ into $\rf{1/\varepsilon} $ intervals of length at
most $\varepsilon h$. Then for some collection $\cI_1,\ldots,\cI_{\nu}$ of these intervals, we have the bound
\begin{equation}
\label{eq:JJ1}
J_{\nu}(\lambda; h)\le \rf{1/\varepsilon}^\nu J^*,
\end{equation}
where $J^*$ is the number of solutions of the congruence
\begin{equation}
\label{eq:ConcentrationProdInt}
(x_1+s)\ldots(x_{\nu}+s)\equiv \lambda \pmod p,\quad x_1\in \cI_1,\ldots,x_{\nu}\in \cI_{\nu}.
\end{equation}
Thus, it suffices to prove the desired bound for $J^*.$

We can assume that $J^*>\nu!$. In particular, we can fix two solutions
$(x_1,\ldots,x_{\nu})=(a_1, \ldots, a_{\nu})$ and $(x_1,\ldots,x_{\nu})=(b_1, \ldots, b_{\nu})$ of~\eqref{eq:ConcentrationProdInt} such that the polynomial
$$
P_0(Z)=(a_1+Z)\ldots(a_{\nu}+Z)-(b_1+Z)\ldots(b_{\nu}+Z)
$$
is not a zero polynomial. Since $1\le a_i, b_i\le h$, this  implies that $P_0(Z)$ is not a zero polynomial modulo $p$. In particular, $P_0(Z)$ is not a constant polynomial.

We note that by the induction hypothesis, the set $(x_1,\ldots, x_{\nu})$ of solutions of the congruence~\eqref{eq:ConcentrationProdInt} for which
$x_i\in\{b_1,\ldots, b_{\nu}\}$ for some $i$,
contributes to $J^*$  at most
\begin{equation}
\label{eq:Bad Sols}
\nu^2 \exp\(c(\nu-1)\frac{\log h}{\log\log h}\) \le
  \exp\(0.5 c(\nu)\frac{\log h}{\log\log h}\),
\end{equation}
provided that $h$ is large enough (and $c(\nu)  > 2c(\nu-1)$).

Consider now the set of polynomials of the form
$$
P(Z)=(x_1+Z)\ldots(x_{\nu}+Z)-(b_1+Z)\ldots(b_{\nu}+Z),
$$
where $(x_1,\ldots,x_{\nu})$ runs through the set of all
solutions of the congruence~\eqref{eq:ConcentrationProdInt} such that
$$
\{x_1,\ldots, x_{\nu}\}\cap \{b_1,\ldots, b_{\nu}\}=\emptyset.
$$
We note that each such  polynomial $P(Z)$ is nonzero and has a form
$$
P(Z)=c_1Z^{\nu-1}+\ldots+c_{\nu-1}Z+c_{\nu},
$$
with $|c_i| \le c_0(\nu)  \varepsilon h^{i}$, $i =1, \ldots, \nu$,
where $c_0(\nu)$ depends only on $\nu$. In particular, since $P(s)\equiv 0\pmod p$, it follows that $P(Z)$ is not a constant polynomial.

Since we have $P(s)\equiv P_0(s)\equiv 0\pmod p$, we see that their resultant $\Res(P,P_0)$
satisfies
\begin{equation}
\label{eq:Res mod p}
\Res(P,P_0) \equiv 0 \pmod p.
\end{equation}
On the other hand, from Corollary~\ref{cor:DeterMagic}, we have that
$$
|\Res(P,P_0)|\le C_0(\nu) \varepsilon h^{\nu^2-1},
$$
with come constant $C_0(\nu)$ that depends only on $\nu$.
Therefore, taking $\varepsilon =  (C_0(\nu)+1)^{-1/(\nu^2-1)}$ we have $|\Res(P,P_0)|<p$,
 which in view of~\eqref{eq:Res mod p} implies that $\Res(P,P_0)=0$.

 Hence, every polynomial $P(Z)$ has a common root with $P_0(Z)$.
 
Let $\beta_1,\ldots,\beta_{n-1}$, $n\le \nu$, be all the roots of $P_0(Z)$. For each
$\beta \in \{\beta_1,\ldots,\beta_{n-1}\} $ we collect together
all solutions $(x_1,\ldots,x_{\nu})$ to~\eqref{eq:ConcentrationProdInt}
for which $P(\beta)=0$. Thus, for some
$\beta \in \{\beta_1,\ldots,\beta_{n-1}\}$ we have
\begin{equation}
\label{eq:JJ2}
J^*\le \exp\(0.5 c(\nu)\frac{\log h}{\log\log h}\)+(\nu-1) J^{**},
\end{equation}
where $J^{**}$ is the number of solutions of the equation
\begin{equation}
\label{eq:ConcentrationProdAlgebraic}
(x_1+\beta)\ldots (x_{\nu}+\beta)=(b_1+\beta)\ldots(b_{\nu}+\beta)
\end{equation}
with $1\le x_i\le h$ such that $x_i\not=b_j.$ This implies, in particular, that the left hand side of~\eqref{eq:ConcentrationProdAlgebraic} is distinct from zero

By Lemma~\ref{lem:HeighDiv}, we conclude that  $\beta $ is an algebraic number of   logarithmic height $O(\log h)$ in an extension $\K$ of $\Q$ of degree $[\K:\Q] \le \nu$.
Now we have that
$$\beta =\frac{\alpha}{q},
$$
where $\alpha$ is an algebraic integer of height at most $O(\log h)$ and $q$ is a positive integer $q\ll h^{\nu}$. From the basic properties of algebraic numbers
it now follows that the numbers
$$qx_i+\alpha, \ i =1, \ldots, \nu, \mand \prod_{i=1}^{\nu}(qb_i+\alpha)
$$
are algebraic integers of $\K$ of height at most $O(\log h)$.

Therefore,
 we conclude that for a sufficiently large $h$ the
 equation~\eqref{eq:ConcentrationProdAlgebraic}
 has at most
\begin{equation}
\label{eq:Good Sols}
\exp\(C(\nu)\frac{\log h}{\log\log h}\) \le
\exp\(0.5 c(\nu)\frac{\log h}{\log\log h}\)
\end{equation}
 solutions, where $C(\nu)$ is the implied constant of  Lemma \ref{lem:Div ANF}
and we also assume that $c(\nu) > 2C(\nu)$. Collecting~\eqref{eq:JJ2}
and~\eqref{eq:Good Sols} together and using~\eqref{eq:JJ1}, we conclude the proof.
\end{proof}

\begin{cor}
\label{cor:ProdSet Fp lin 1}  Let $\nu \ge 2$ be a fixed integer. Assume
that for some sufficiently large positive integer $h$ and prime
$p$ we have
$$
h<p^{1/(\nu^2-1)}.
$$
For  $s \in \F_p$ we consider the set
$$
\cA = \left\{ x+s ~:~ 1\le x\le  h \right\} \subseteq
\F_p.
$$
Then
$$
\# (\cA^{(\nu)}) > \exp\(-c(\nu) \frac{\log h  }{ \log \log h }\) h ^{\nu},
$$
where  $c(\nu)$ depends only on $\nu$.
\end{cor}

We now obtain similar results for the set of fractions $(x+s)/(x+t)$.

\begin{lemma}
\label{lem:ProdSet Fp birat 1}  Let $\nu \ge 1$ be a fixed integer. Assume
that for some sufficiently large positive integer $h$ and prime
$p$ we have
$$
h < p^{c \nu^{-4}},
$$
where $c$ is a certain absolute constant. For pairwise distinct
$s,t \in \F_p$ we consider the set
$$
\cA = \left\{\frac{x+s}{x+t}~:~ 1\le x\le  h \right\} \subseteq
\F_p.
$$
Then
$$
\# (\cA^{(\nu)}) > \exp\(-c(\nu) \frac{\log h  }{\sqrt{\log \log h }}\) h ^{\nu},
$$
where  $c(\nu)$ depends only on $\nu$.
\end{lemma}

\begin{proof}  We consider the collection $\cP\subseteq \Z[Z_1,Z_2]$
of   polynomials
$$
P_{\vec{x},\vec{y}}(Z_1,Z_2) =
\prod_{i=1}^\nu (x_i + Z_1)  \prod_{j=1}^\nu (y_j + Z_2)
-  \prod_{i=1}^\nu (x_i + Z_2)   \prod_{j=1}^\nu (y_j + Z_1),
$$
where $\vec{x} = (x_1, \ldots, x_\nu)$ and
$\vec{y} = (y_1, \ldots, y_\nu)$ are integral vectors with
entries in $[0,h]$ and such that
$$
P_{\vec{x},\vec{y}}(s,t) \equiv 0 \pmod p.
$$

As in the proof of Lemma~\ref{lem:GProdSet Fp lin 1}, we can assume that
 $\cP$ contains non-zero polynomials.

Clearly,  every $P\in \cP$ if of degree at most $2\nu-1$
and of logarithmic height at most $3\nu \log h $.

We take a family $\cP_0$ containing the largest possible number
$$
N \le (\nu +1)^2 - 1
$$
of linearly independent polynomials $P_1, \ldots, P_N \in \cP$,  and consider the
variety
$$
\cV: \ P_1(Z_1,Z_2) = \ldots =P_N(Z_1,Z_2) = 0.
$$

We claim that $f(Z_1,Z_2) = Z_1-Z_2$  does not vanish on $\cV$.

Indeed, if  $f(Z_1,Z_2)$ vanishes on $\cV$ then by
Lemma~\ref{lem:Hilb} we see that there are polynomials
$Q_1, \ldots, Q_N \in \Z[Z_1,Z_2]$ and  positive integers
$b$ and $r$ with
\begin{equation}
\label{eq:b small birat}
\log b \le c_0 \nu^{3}(\nu \log h  + \nu) \le 2c_0 \nu^{4} \log h
\end{equation}
for some absolute constant $c_0$ (provided that $h$ is large enough)
and such that
$$
P_1Q_1+ \ldots + P_NQ_N = b(Z_1-Z_2)^r.
$$
Substituting $(Z_1,Z_2) = (s,t)$ and using that $s$ and $t$ are disctinct
elements of $\F_p$ we obtain $p \mid b$. Taking $c = 1/(2c_0+1)$ in
the condition of the theorem, we see from~\eqref{eq:b small birat} that
this is impossible.

Hence for the set
$$
\cU = \cV \cap [Z_1-Z_2 \ne 0]
$$
is nonempty. Applying Lemma~\ref{lem:SmallZero}
we see that  it has a point $(\beta_1, \beta_2)$
with components of   logarithmic height $O(\log h)$
in an extension $\K$ of $\Q$ of degree $[\K:\Q] = O(1)$.

Let $\cI = \{0, 1, \ldots, h \}$. Consider the maps
$\Phi:\  \cI^\nu  \to \F_p$
given by
$$
\Phi: \ \vec{x} = (x_1, \ldots, x_\nu) \mapsto \prod_{j=1}^\nu \frac{x_j+s}{x_j+t}
$$
and $\Psi:  \cI^\nu  \to \K$
given by
$$
\Psi: \ \vec{x} = (x_1, \ldots, x_\nu) \mapsto \prod_{j=1}^\nu \frac{x_j+\beta_1}{x_j+\beta_2}.
$$
By construction of  $(\beta_1, \beta_2)$ we have that
$\Psi(\vec{x}) = \Psi(\vec{y})$ if $\Phi(\vec{x}) =
\Phi(\vec{y})$.
Hence
$$
\#(\cA^{(\nu)}) \ge \Im \Psi = (\# \cC^{(\nu)}),
$$
where $\Im \Psi$ is the image set of the map $\Psi$ and
$$
\cC = \left\{\frac{x+\beta_1}{x+\beta_2}~:~ 1\le x\le h \right\}
\subseteq \K.
$$
Using Corollary~\ref{cor:ProdSet} derive the result.
\end{proof}

\subsection{Shifted Sets in Conjugacy Classes of $\cG_e$}

We are now able to present our main technical tools.

\begin{lemma}
\label{lem:TwoCongr} Let  $\alpha$, $\beta$,
$\delta$, $\zeta$, $\cI$ and  $\cS$ be as in Lemma~\ref{lem:TripleSums}.
For $x \in \F_p$ we define
$$
r(x) = \max_{A_0, A_1 \in \F_p}
\# \{t\in \cS~:~ (t+\zeta^\nu x)^e \equiv A_\nu \pmod p, \ \nu = 0,1\}.
$$
Then for  $e \le p^{1-\delta}$ we have
$$
\min_{x\in \cI} r(x) \ll p^{-\xi}\#\cS,
$$
where $\xi > 0$ depends only on $\delta$.
\end{lemma}

\begin{proof}
Clearly $r(x) \le r_0(x) + r_1(x)$,
where
$$
r_\nu(x) =  \max_{A_\nu \in \F_p}
\# \{t\in \cT_\nu~:~ (t+\zeta^\nu x) \equiv A_\nu \pmod p\}
$$
and $\cT_0$ and $\cT_1$ are  as in Lemma~\ref{lem:TripleSums}.
Let $d=(p-1)/e$.
We denote by $\chi_0$ the principal character modulo $p$
and by $\chi_1,\ldots,\chi_{d-1}$ the  other characters with
$\chi_j^d=\chi_0$.
Then, using the orthogonality of multiplicative
characters (see~\cite[Section~3.1]{IwKow}), we write
$$
r_\nu(x) =  \frac{1}{d}   \sum_{t \in \cT_\nu} \sum_{j=0}^{d-1}
\chi_j(t+\zeta^\nu x)\overline \chi_j(A_\nu).
$$
Thus, for $\nu =0,1$,
\begin{equation*}
\begin{split}
\sum_{x\in \cI} r_\nu(x)^2  &
 \le \frac{1}{d} \sum_{j=0}^{d-1}\left| \sum_{x\in \cI}
 \sum_{t_1,t_2\in \cT_\nu}  \chi_j(\zeta^\nu x+t_1)\overline \chi_j(\zeta^\nu x+t_2)\right|.
\end{split}
\end{equation*}
The contribution of the principal character $\chi_0$ is $\#\cI (\#S)^2$.
Furthermore, by Lemma~\ref{lem:TripleSums} the contribution from any nonprincipal character
is $\#\cI (\#S)^2  p^{-\eta}$. Therefore,
$$
\sum_{x\in \cI} r_\nu(x)^2  \ll  \frac{e}{p-1} \#\cI (\#S)^2
+\#\cI (\#S)^2  p^{-\eta}, \qquad \nu =0,1,
$$
which concludes the proof.
\end{proof}

We also see that Corollary~\ref{cor:ProdSet Fp lin 1} yields:

\begin{lemma}
\label{lem:ProdSet Fp lin 2} Let $\delta >0$ be fixed.
Let $\cA$ be as in Corollary~\ref{cor:ProdSet Fp lin 1}.
If $\cA \subseteq r\cG_e$ where $r\in\F_p^*$ and $e < p^\delta$
then,
$$
h =O\( e^{c_0\delta}\)
$$
where $c_0$ is some absolute constant.
\end{lemma}

Finally, we immediately derive from Lemma~\ref{lem:ProdSet Fp birat 1}:

\begin{lemma}
\label{lem:ProdSet Fp birat 2} Let $\delta >0$ be fixed.
Let $\cA$ be as in Lemma~\ref{lem:ProdSet Fp birat 1}.
If $\cA \subseteq \cG_e$ where $e < p^\delta$
then,
$$
h =O\( e^{c_0\delta^{1/3}}\)
$$
where $c_0$ is some absolute constant.
\end{lemma}

\section{Main Results}

\subsection{Hidden Shifted Power Problem}

Here we give deterministic and probabilistic algorithms for the
Hidden Shifted Power Problem that runs in about the same time as
the interpolation algorithm,  but use significantly less oracle
calls.

%For our deterministic  algorithm we have to assume that
%$\ell$-th power nonresidues, for all prime divisors $\ell\mid e$,
%are given. In particular, it happens if a primitive root
%$g \in \F_p$ is given.

\begin{theorem}
\label{thm:Small_e Oracle} For a prime $p$ and a positive integer
$e\mid p-1$ with $e\le p^{1-\delta}$, given an oracle $\Oes$ for
some unknown $s\in S_0$ with a known $S_0\subseteq\F_p$, $\#S_0\le e$,
there is a deterministic algorithm that for any fixed $\eps>0$
makes $O(1)$ calls to the oracle $\Oes$ and finds $s$ in time
$e^{1+\eps}(\log p)^{O(1)}$.
\end{theorem}

\begin{proof} Let $e = p^\rho$. First we consider the case of large $e$
when $\rho \ge 0.65$.

We fix some integer $m \ge 3$
so that $p$ and $e$ satisfy the condition of
Lemma~\ref{lem:ShkVyu}. We now make $m$ calls to $\Oes$
with $j= 1, \ldots, m$, getting $A_j = (s+j)^e$.

We now take a set $S_m$ that consists of all elements $t \in S_0$, for which
$$
(t+j)^e = A_j, \qquad j= 1, \ldots, m.
$$
Thus, $S_{m}$  is the set of candidates for $s$ after $m$
calls. To find $S_{m}$, we can test all
elements $t\in S_0$.
This requires the running time $e(\log p)^{O(1)}$.
 Clearly, there are some
$a_j \in \F_p^*$, $j = 1, \ldots, m$, so that
$$
s \in S_{m}\subseteq S_0 \cap\(\bigcap_{j = 1}^m \(a_j\cG_e - j\)\).
$$
By Lemma~\ref{lem:ShkVyu} we see that $\#S_{m} =
O(e^{m/(2m-1)})$. 
The second part of our algorithm is iterative which starts with the set $S=S_{m}$
with $s\in S$
described in the above with an appropriate choice of $m = O(1)$
so that it is
of cardinality $\# S \le e^{1/2 + \varepsilon}$
(which can be constructed after $O(1)$ calls),
and then at each step it makes  a call to $\Oes$ so that
after its reply we get a substantially
smaller set of candidates.

More precisely, denote
$$
\vartheta = \frac{1}{4}\(3 +\rho - \sqrt{1 + \rho^2}\)
$$
and assume that at some stage we are given a set
$S \subseteq \F_p$ with  $s\in S$ of cardinality
$p^{0.05}<\# S \le e^{1/2 + \varepsilon}$.
 We show how to make a call to
$\Oes$ so that after its reply we get
a set of candidates of the size
reduced by a small power of $p$.

Define $\alpha$  by $\#S = p^{\alpha}$ and note that for any
\begin{equation}
\label{eq:alphabeta}
\beta > \frac{1}{4}\(3-2 \alpha - \sqrt{1 + 4 \alpha^2}\)
\end{equation}
and an appropriate $\delta > 0$ the condition of Lemma~\ref{lem:TripleSums} is satisfied
so Lemma~\ref{lem:TwoCongr}  applies.

Take
$$
h =\rf{ p^{\beta}}
$$
and for all $t \in S$ and  $x \in [0, h]$  compute the pairs
$\((t+ x)^e, (t+\zeta x)^e\)$, where $\zeta$ is as
in  Lemma~\ref{lem:TripleSums}. We now order, for each $x$, the list of pairs
in the ascending order  with respect to the first component
and then with respect to the second component.
Scanning this ordered list we find $x$ that satisfies the bound
of  Lemma~\ref{lem:TwoCongr}.
 We use this $x$ for
the next two calls  to get $A=(x+s)^e$ and $B =(\zeta x+s)^e$.
Then the new set of the
candidates  is defined as
$$T=\{t\in S~:~(t+x)^e=A,\  (t+\zeta x)^e=B\}.$$
Clearly, we have
\begin{equation}
\label{eq:small T}
\#T\ll\#S p^{-\xi}
\end{equation}
for some $\xi>0$ that depends only on $\alpha$ and $\beta$.

The total cost
of this step is $p^{\beta+o(1)} \# S$.
Since $\beta$ is an arbitrary number satisfying~\eqref{eq:alphabeta},
we see that it is possible to find this
set in time $O\(p^{\(3+2 \alpha - \sqrt{1 + 4 \alpha^2}\)/4+ \eta }\)$
for  an arbitrary $\eta>0$.
Since the above exponent is a monotonically increasing function
of $\alpha$ and $\alpha < \rho/2 + \varepsilon$ we see that the cost
of each step can be made at most $p^{\vartheta+0.03}$
provided that $\varepsilon$ is small enough.

The procedure terminates when we get the set $S$ of candidates with
$\# S\le p^{0.05}$.
It is obvious that~\eqref{eq:small T}
implies that  the procedure terminates after $O(1)$ steps and
has the time compelxity $e(\log p)^{O(1)} +  O\(p^{\vartheta+0.03}\)$.
Since $\vartheta < \rho-0.03$ for $\rho \ge 0.65$, the total complexity
of the above procedure is $O(e)$.

The final part of our algorithm is also iterative which starts with
the set $S$ with $\# S\le p^{0.05}$. We take
$$
h = \rf{e^{0.56}} \mand  \cI = [0,h).
$$

For $x \in \F_p$ we define
$$
R(x)=\#\left\{(s_1,s_2)\in S\times S~:~s_1\neq s_2,\
\frac{x+s_1}{x+s_2}\in\cG_e\right\}.
$$
Clearly
\begin{equation}
\label{eq:R and Q}
\sum_{x\in \cI} R(x) =  Q,
\end{equation}
where
$$
Q=\#\left\{(x,s ,t)\in \cI  \times S\times S ~:~
s\neq t,\ \frac{x+s }{x+t}\in\cG_e\right\}.
$$
We write
\begin{equation}
\label{eq:Q and Qst}
Q = \sum_{\substack{s,t \in S\\s\neq t}}   Q(s,t),
\end{equation}
where
$$Q(s,t) = \#\cQ(s,t)
\mand  \cQ(s,t) = \left\{x \in \cI
~:~ \frac{x+s }{x+t}\in\cG_e\right\}.
$$

As in the proof of Lemma~\ref{lem:TwoCongr}  we
put $d=(p-1)/e$,
denote by $\chi_0$ be the principal characters modulo $p$
and by $\chi_1,\ldots,\chi_{d-1}$ the  other characters with
$\chi_j^d=\chi_0$. We have
$$
Q(s,t) = \frac1d\sum_{x\in\cI}\sum_{j=0}^{d-1}\chi_j\(\frac{x+s}{x+t}\).
$$
Using Lemma~\ref{lem:Weil3} we get
$$Q(s,t) \le \frac{h}d + O\(p^{1/2}\log p\).$$
The substitution in~\eqref{eq:Q and Qst} and then
using~\eqref{eq:R and Q} implies
$$
\sum_{x\in \cI} R(x) \le \(\# S\)^2 h \(\frac{e}{p-1}
+ O\(\frac{p^{1/2}\log p}{h}\)\).
$$
Therefore, we see there is
$x\in\{0,\ldots,h-1\}$such that
\begin{equation}
\label{welldistrx1}
R(x)\le  \(\# S\)^2\(\frac{e}{p-1} + O\(\frac{p^{1/2}\log p}{h}\)\).
\end{equation}
We can consider that $\delta\le 0.05$.
By the supposition on $e$ and the choice of $h$ we get
$$R(x)\ll  \(\# S\)^2 p^{-\delta}.$$

To find the desired value of $x$ for which~\eqref{welldistrx1} holds,
we  simply compute $(x+t)^e$ for all $x=0,\ldots, h-1$
and $t \in S$ in time
$h\# S (\log p)^{O(1)} \ll p^{0.62} \le e.$

We now use any $x$ that satisfies~\eqref{welldistrx1} for
the next call and get $A=(x+s)^e$. Then the new set of the
candidates  is defined as
$$T=\{t\in S~:~(x+t)^e=A\}.$$
Clearly, we have
$$
\#T \ll R(x)^{1/2} + 1 \ll \#S p^{-\delta/2} + 1.
$$

We now repeat the same with $T$ instead of $S$ and search
for a new appropriate value of $x$.

Thus in $O(1)$ steps  this procedure  produces a set $T$ of cardinality
$\#T = O(1)$ with $s \in T$. Checking whether $s = t$ for every element
$t\in T$ takes at most $\# T = O(1)$ calls to $\Oes$ with $x = -t$ with
$t \in T$, until $\Oes$ returns  zero.
This completes the proof when $\rho \ge 0.65$.

Finally, to prove  the result for $\rho < 0.65$,
we again start the algorithm  with $O(1)$ calls to produce a
set $S$
such that $s\in S$ and $\# S \le e^{1/2 + \varepsilon/4}$.

We now take
$$
h = \rf{e^{1/2+\varepsilon/2}} \mand  \cI = [0,h).
$$
Next, we define $R(x), Q, Q(s,t), \cQ(s,t)$ as in the previous case.

Denote
$$\cQ(s,t)\times\cQ(s,t)=\{(x,y)~:~x\in\cQ(s,t),\,y\in\cQ(s,t)\}.$$
Clearly
$$
\#(\cQ(s,t)\times\cQ(s,t)) = Q(s,t)^2.
$$
Note that if
$$
\frac{(x+s)(y+s)}{(x+t)(y+t)}  = 1
$$
then (since $s \ne t$) for each $x \in  \cQ(s,t) $ there is at most
one value of $y  \in \{0,\ldots,h-1\}$.
So such solutions contribute at most $Q(s,t)$ to $\cQ(s,t)\times\cQ(s,t)$.
Thus
\begin{equation}
\label{eq: Q2-Q}
Q(s,t)^2 - Q(s,t) \le \#\left\{x , y \in \cI
~:~ \frac{(x+s)(y+s)}{(x+t)(y+t)}\in\cG_e\setminus \{1\}\right\}.
\end{equation}
If for some $a \in\cG_e\setminus \{1\}$, we have
\begin{equation}
\label{eq: Eq1}
\frac{(x+s)(y+s)}{(x+t)(y+t)}  = a
\end{equation}
then we can write~\eqref{eq: Eq1} in the form
$$(a-1)xy+(at-s)(x+y)+(at^2-s^2)=0,
$$
or
$$
(x  + u)(y +u) = v,
$$
where
$$
u = \frac{at-s}{a-1} \mand v =  \frac{s^2-at^2}{(a-1)}+u^2 .
$$
Since $v \in \F_p^*$, using Lemma~\ref{lem:CillGar}, we see that the equation~\eqref{eq: Eq1}
has at most $h^{3/2+o(1)} p^{-1/2} + h^{o(1)}$ solutions.
We now see from~\eqref{eq: Q2-Q}
\begin{equation}
\label{eq: Qst}
Q(s,t)^2   \le e\(h^{3/2+o(1)} p^{-1/2} + h^{o(1)}\),
\end{equation}
%%IS
(here and throughout the proof we write $o(1)$ for a quantity that
tends to zero provided that $e\to\infty$).

Furthermore, under the assumption that $\rho < 0.65$,
taking a sufficiently small $\varepsilon$,
we have  $h \le p^{1/3}$.
Therefore the  bound~\eqref{eq: Qst}  simplifies as
$$
Q(s,t) \le e^{1/2}  h^{o(1)}.
$$
Therefore, the substitution in~\eqref{eq:Q and Qst} and then
using~\eqref{eq:R and Q} implies
$$
\sum_{x\in \cI} R(x) \le \(\# S\)^2 e^{1/2}   h^{o(1)}  .
$$
Thus, recalling the definition of $h$, we see there is $x\in\{0,\ldots,h-1\}$
such that
\begin{equation}
\label{welldistrx}
R(x)\le  \(\# S\)^2  e^{1/2}   h^{-1+o(1)}
=   \(\# S\)^2   e^{-\varepsilon/2+o(1)} \ll  \(\# S\)^2 e^{-\varepsilon/3}.
\end{equation}

To find the desired value of $x$ for which~\eqref{welldistrx} holds,
we  simply compute $(x+t)^e$ for all $x=0,\ldots, h$
and $t \in S$ in time
$h\# S (\log p)^{O(1)} \le  e^{1 + \varepsilon} (\log p)^{O(1)}.$

We now use any $x$ that satisfies~\eqref{welldistrx} for
the next call and get $A=(x+s)^e$. Then the new set of the
candidates  is defined as
$$T=\{t\in S~:~(x+t)^e=A\}.$$
Clearly, we have
$$
 \#T \ll \#S e^{-\xi} + 1
$$
for some $\xi >0$ that depends only on $\varepsilon$.

We now repeat the same with $T$ instead of $S$ and search
for a new appropriate value of $x$.

Thus in $O(1)$ steps  this procedure  produces a set $T$ of cardinality
$\#T = O(1)$ with $s \in T$. Checking whether $s = t$ for every element
$t\in T$ takes at most $\# T = O(1)$ calls to $\Oes$ with $x = -t$ with
$t \in T$, until $\Oes$ returns  zero.
This completes the proof.
\end{proof}

\begin{cor}
\label{cor:Small_e Oracle1} For a prime $p$ and a positive integer
$e\mid p-1$ with $e\le p^{1-\delta}$, given an oracle $\Oes$ for
some unknown $s\in \F_p$
and $\ell$-th power nonresidues
for all prime divisors $\ell\mid e$,
there is a deterministic algorithm that for any fixed $\eps>0$
makes $O(1)$ calls to the oracle $\Oes$ and finds $s$ in time
$e^{1+\eps}(\log p)^{O(1)}$.
\end{cor}

\begin{proof}
We make the first call to $\Oes$ with $x = 0$,
getting $A_0 = s^e$. If $A_0=0$ then $s=0$ and we are done. Now
assume that $A_0\neq0$.

We see from Lemma~\ref{lem:BinEq} that  we can construct
the set
\begin{equation}
\label{eq:Set S0}
S_0=\left\{t~:~t^e = A_0 \right\}
\end{equation}
of candidates for $s$ in time $e(\log p)^{O(1)}$.
Now it suffices to use Theorem~\ref{thm:Small_e Oracle}.
\end{proof}

%Recall that   for every prime $\ell\mid p-1$,
%the smallest $\ell$-th power nonresidue modulo $p$ is
%at most $p^{1/4\e^{1/2}+o(1)}$ (uniformly over $\ell$),
%where $\e=2.7182\ldots$ is
%the base of the natural logarithms, see, for example,~\cite{Burth}.
%Trivially, $e$  can be factored in  time $e^{1/2 + o(1)}$,
%Hence, we can find  $\ell$-th power nonresidues modulo $p$ for every
%prime $\ell \mid e$  in
%time $p^{1/4\e^{1/2}+o(1)}+e^{1/2 + o(1)}$.
%Therefore we have:

\begin{cor}
\label{cor:Small_e Oracle2}
For a prime $p$ and a positive integer
$e\mid p-1$ with $e\le p^{1-\delta}$, given an oracle $\Oes$ for
some unknown $s\in \F_p$,
there is a deterministic algorithm  that for any fixed $\eps>0$ makes $O(1)$ calls to
the oracle $\Oes$ and finds $s$ in time
%% $O\(e^{1+\eps} +p^{\eps}\)$.
$O\(ep^{\eps}\)$.
\end{cor}

\begin{proof}
We can consider that $\eps<1/2$.
Trivially, $e$  can be factored in  time $e^{1/2 + o(1)}$.
For any prime $\ell\mid e$ we take $\alpha_\ell$ so that
$\ell^{\alpha_\ell}\|p-1$. Denote $y=\fl{p^\eps}$. For any $x=1,\ldots,y$
we take $\gamma_\ell(x)$ as the largest nonnegative integer $\gamma\le\alpha_\ell$ 
so that
$$x^{(p-1)/\ell^{\gamma}}\equiv1\pmod p.$$
Next, we denote
$$\gamma_\ell=\min\{\gamma_\ell(x)~:~1\le x\le y\}$$
and for any prime $\ell\mid e$ we choose $x=x(\ell)$ so that $\gamma_\ell(x)=\gamma_\ell$.
Let
$$n=\prod_{\substack{\ell\mid e\\\ell~\mathrm{prime}}}\ell^{\gamma_\ell}.$$
We have $x^{(p-1)/n}\equiv1\pmod p$ for all $x=1,\ldots,y$. Therefore,
from Corollary~\ref{cor:genset} we deduce that
$n \le (1/\eps)^{c/\eps}$ for some absolute constant $c$. The running time for finding $n$ and all $x(\ell)$
is $p^\eps(\log p)^{O(1)}$. Using Lemma~\ref{lem:list_cand}, we find a set
$S_0$ of candidates for $s$ of cardinality at most $e$ in time
$e (\log p)^{O(1)}n^{O(1)}$. By Theorem~\ref{thm:Small_e Oracle}, we find $s$
in time $O\(\(e^{1+\eps} +p^{\eps}\)(\log p)^{O(1)}\)$. Replacing $\eps$
with $\eps/2$, we get the running time
$$O\(\(e^{1+\eps/2} +p^{\eps/2}\)(\log p)^{O(1)}\)
= O\(e^{1+\eps} +p^{\eps}\) = O(e p^{\eps})
$$
as required.
\end{proof}

More precisely,  it is easy to  see that in the algorithm 
of Corollary~\ref{cor:Small_e Oracle2} the number of calls to the oracle $\Oes$ needed to find $S_0$ and the running time for this step are
bounded by 
$$
(1/\eps)^{c/\eps} \mand p^\eps(\log p)^{O(1)} + e(\log p)^{O(1)}(1/\eps)^{O(1/\eps)},
$$
respectively, where $c$ is an absolute constant. 

We note that   the  Extended Riemann Hypothesis implies that for
the smallest $\ell$-th power nonresidue modulo $p$ is
 $O((\log p)^2)$ (uniformly over primes $\ell\mid p-1$),
see~\cite[Chapter~9, Corollary~1]{Mont}.
Hence, we obtain:

\begin{cor}
\label{cor:Small_e Oracle-ERH} Assuming the Extended
Riemann Hypothesis, for a prime $p$ and a positive integer
$e\mid p-1$ with $e\le p^{1-\delta}$, given an oracle $\Oes$ for
some unknown $s\in \F_p$,
there is a deterministic algorithm  that for any fixed $\eps>0$ makes $O(1)$ calls to
the oracle $\Oes$ and finds $s$  in time  $e^{1+\eps} (\log p)^{O(1)}$.
\end{cor}

%Moreover, the statement of this corollary holds under a weaker supposition that
%all nontrivial zeros of all $L$-functions $L(s,\chi)$ for all characters
%$\chi$ modulo $p$ satisy $\Re s\le 1-\delta$, see~\cite[Chapter~9,
%Theorem~1]{Mont}. A further relaxation of supposition on zeros to
%$(1-\Re s)\log p\to\infty$ implies the estimate $e^{1+\eps}p^{o(1)}$
%for the running time.

We also note that by a result of Burgess and Elliott~\cite{BurgEll}
for almost all primes $p$ the smallest primitive
root is $(\log p)^{2+\varepsilon}$ for any $\varepsilon > 0$,
see also~\cite{EllMur}. Thus for almost all primes we have an
unconditional version of Corollary~\ref{cor:Small_e Oracle-ERH}.

We now present a probabilistic algorithm which is slightly more efficient in
some cases.

\begin{theorem}
\label{thm:Small_e Oracle Rand} For a prime $p$ and a positive integer
$e\mid p-1$ with $e\le p^{1-\delta}$, given an oracle $\Oes$ for
some unknown $s\in \F_p$, there is
 a probabilistic algorithm
that makes in average
$O(1)$ calls to the oracle $\Oes$ and
finds $s$ in the expected time $e(\log p)^{O(1)}$
\end{theorem}

\begin{proof} We again start
from the first call to $\Oes$ with $x = 0$.
Using a probabilistic algorithm, we can find the set $S_0$
given by~\eqref{eq:Set S0} in the expected
time $e(\log p)^{O(1)}$, see~\cite[Corollary~14.16]{vzGG}.
Then we make next calls with random $x_1,\ldots,x_\nu$ where
$$\nu=\fl{\frac{3\log p}{\log(p/e)}}+1.$$
For any $j$ and any $s_1\neq s_2$ the probability of the event
\begin{equation}
\label{eq:even j}
(x_j+s_1)/(x_j+s_2)\in\cG_e
\end{equation}
is at most $e/p$. Hence, by the choice of $\nu$, the probability
of the event $(x_j+s_1)/(x_j+s_2)\in\cG_e$ for all $j$ is at most
$(e/p)^{\nu}<p^{-3}$. Next, the probability that for some $s_1\neq
s_2$ we have~\eqref{eq:even j} for all $j$ is at most $1/p$.
Therefore, the random choice of $x_1,\ldots,x_\nu$ determines $s$
with high probability. To find $s$ we have to test elements from
$S_0$. This can be done in time $e(\log p)^{O(1)}$, and the result
follows.
\end{proof}

The following result is applicable to the case when $e$ does not satisfy
the restriction in Theorem~\ref{thm:Small_e Oracle} (namely, to
$e=p^{1+o(1)}$ as $p\to\infty$).

\begin{theorem}
\label{thm:Large_e Oracle} For a prime $p$ and a positive integer
$e\mid p-1$, given an oracle
$\Oes$ for some unknown $s\in \F_p$, there is a deterministic
algorithm that makes $O\(\log p/(\log(p/e))\)$ calls to the oracle
$\Oes$ and finds $s$ in time $p(\log p)^{O(1)}$.
\end{theorem}

\begin{proof} For $e\le p^{0.9}$ the result follows immediately from
Theorem~\ref{thm:Small_e Oracle}. We now assume that $e>p^{0.9}$.

Again, we fix some integer $m \ge 1$ and  now make $m$ calls
to $\Oes$ with $j = 1, \ldots, m$, getting $A_j = (s+j)^e$.
If $A_j=0$ for some $j$, then $s=-j$ and we are done. Hence we can
assume  that $A_j\neq0$ for $j=1,\ldots,\mu$. Our immediate aim is
to estimate the cardinality of the set $S_m$ of candidates after
$m$ calls:
$$S_m=\{x \in \F_p~:~(x+j)^e=a_j,\,j=1,\ldots,m\}.$$

As in the proof of Lemma~\ref{lem:TwoCongr}  we
put $d=(p-1)/e$,
denote by $\chi_0$ be the principal characters modulo $p$
and by $\chi_1,\ldots,\chi_{d-1}$ the  other characters with
$\chi_i^d=\chi_0$. The condition $x^e=A_j$ determines the values
$\chi_i(x)=a_{i,j}$. We have
$$
\# S_m=d^{-m}\sum_{i_1,\ldots,i_m=0}^{d-1}\prod_{j=1}^m
\chi_{i_j}(x+j)\overline{a_{{i_j},j}}.
$$
Applying Lemma~\ref{lem:Weil2} we get
$$\# S_m=d^{-m}(p-m)+O(m p^{1/2}).$$
Setting
$$m=\fl{\frac{\log p}{2\log(p-1)/\log e}}+1,$$
we have $\# S_m\le p^{1/2+o(1)}$. We need the running time
$p(\log p)^{O(1)}$ to find $S_m$.

Now we proceed as in the proof of
Theorem~\ref{thm:Small_e Oracle}  with
$$h=\fl{\frac pe p^{1/2}(\log p)^2}.$$
After the $j$-th call, $j\ge m$, we get the set $S=S_j$ of candidates
for $s$. Next, we define $R(x), Q, Q(s,t), \cQ(s,t)$ as
in Theorem~\ref{thm:Small_e Oracle}. Using~\eqref{welldistrx1}
we get the new set $S=S_j$ of candidates for $s$ with
$$\#S_{j+1}\le\max\left\{1, (1+o(1))(e/p)^{1/2} \# S_j\right\}.$$
Thus in $\ll\log p/(\log p/\log e)$ steps  this procedure  produces a set
$T$ of cardinality
$\#T = O(1)$ with $s \in T$. Checking whether $s = t$ for every element
$t\in T$ takes at most $\# T = O(1)$ calls to $\Oes$ with $x = -t$ with
$t \in T$, until $\Oes$ returns  zero. Since the time to find
$S_{j+1}$ is $(\log p)^{O(1)}h \# S_j\le p(\log p)^{O(1)}$,  we obtain
the desired result.
\end{proof}

Combining Corollary~\ref{cor:Small_e Oracle1} and Theorem~\ref{thm:Large_e Oracle}
we get the following result:
%Combining Theorems~\ref{thm:Small_e Oracle} and~\ref{thm:Large_e Oracle}
%we get the following.

\begin{cor}
\label{thm:all_e Oracle} For a prime $p$ and a positive integer
$e\mid p-1$, given an oracle $\Oes$ for
some unknown $s\in \F_p$,
and $\ell$-th power nonresidues
for all prime divisors $\ell\mid e$,
there is a deterministic algorithm that for any fixed $\eps>0$
makes $O\(\log p/(\log(p/e))\)$ calls to the oracle $\Oes$ and
finds $s$ in time $e^{1+\eps}(\log p)^{O(1)}$.
\end{cor}

\subsection{Shifted Power Identity Testing with Known $t$}

\begin{theorem}
\label{thm:Medium_e_s} For a prime $p$ and a positive integer
$e\mid p-1$ with $e\le p^{1-\delta}$ for some fixed $\delta >0$,
given an oracle $\Oes$ for some unknown $s\in \F_p$ and $t\in
\F_p$, there is a deterministic algorithm to decide whether $s =
t$ in time $e^{1/4+o(1)}(\log p)^{O(1)}$ as $e\to\infty$.
\end{theorem}

\begin{proof}
For  integers $a$ and $H$ with $0 \le a < a+H <p$,
we consider the interval $\cI = [a+1,a+H]$ of size $H<p^{1/3}$.

Fix some integer $m \ge 1$ so that $p$ and $e$ satisfy the condition
of Lemma~\ref{lem:ShkVyu}. We put $\ell = m!$,
$\ell_s = m!/(s+1)$, $s=1, \ldots, m-1$, and $K = \fl{H/\ell}$.

Let $\cJ = \{a+\ell, \ldots, a + \ell K \}$. Thus $\cJ \subseteq
\cI$. Let $\cA = \cJ/\cJ$, that is,
$$
\cA = \{j_1/j_2~:~j_1,j_2 \in \cJ\}.
$$
By Lemma~\ref{lem:ACZ} we see that
$$
\frac{a+\ell h} {a+\ell i} = \frac{a+\ell j} {a+\ell k},
\quad i, j, h, k\in [1,H],
$$
has $H^{2 +o(1)}$ solutions as $H\to\infty$. Therefore,
\begin{equation}
\label{eq:A bound} \#\cA \ge H^{2+o(1)}.
\end{equation}

Next we observe that
$$\cA + s \subseteq \{(s+1)u~:~u \in \cI/\cI\},
$$
since
$$
\frac{a+ \ell h} {a+ \ell i}+ s = (s+1) \frac{a + s\ell_s i+
\ell_s h}{a+\ell i}.
$$
and $s\ell_s i+ \ell_s h \le (s+1) \ell_s K \le H$.

Clearly if $\cI \in r\cG_e$ then $\cA\subseteq\cG_e$ and
$\cA+s\subseteq (s+1)\cG_e$. The system of equations
$$x_0+ s=x_s,\quad x_s\in(s+1)\cG_e,\quad s=0,\ldots,m-1,$$
has at least $\#\cA$ solutions of the form
$x_0\in\cA$, $x_s = x_0 + s$, $s =1, \ldots,m$. We now set
$$
H =\fl{e^{1/4 + \varepsilon}}
$$
for a sufficiently small $\varepsilon > 0$. By
Lemma~\ref{lem:ShkVyu} we have $\#\cA\ll e^{(m+1)/(2m+1)}$
which, for a sufficiently large $m$ and the above choice of $H$,
contradicts~\eqref{eq:A bound}. Since  $\varepsilon > 0$ is
arbitrary, we now complete the proof by simply choosing
$\cY =[1,H]$ and recalling~\eqref{eq:Cond Fin}.
\end{proof}

For large
values of $e$ we can use bounds of character sums.

\begin{theorem}
\label{thm:Large_e_s} For a prime $p$ and a positive integer
$e\mid p-1$ with $e\le (p-1)/2$, given an oracle $\Oes$ for some
unknown $s\in \F_p$ and $t\in \F_p$, there is a deterministic
algorithm to decide whether $s = t$ in time $ p^{1/4+ o(1)}$
as $p\to\infty$.
\end{theorem}

\begin{proof}  We argue as in the proof of Theorem~\ref{thm:Medium_e_s}.
Recalling~\eqref{eq:Cond Fin}, we see that for any
multiplicative character $\chi$ of order $e$ of $\F_p^*$ we have
$$
\left|\sum_{y \in \cY} \chi(y - r) \right |  = \# \cY.
$$
 We now fix a sufficiently small $\varepsilon> 0$
and take   $\cY = \{1, \ldots, h\}$
where  $h = \rf{ p^{1/4+\varepsilon}}$. Applying Lemma~\ref{lem:PVB}
with a large enough $\nu$, we obtain a contradiction.
Since $\varepsilon> 0$ is arbitrary, the result now follows.
\end{proof}

Collecting the results of Theorems~\ref{thm:Medium_e_s}
and~\ref{thm:Large_e_s}, we obtain an algorithm of complexity
$e^{1/4} p^{o(1)}$ for any $e\le (p-1)/2$.

For small values of $e$ we can use Lemma~\ref{lem:ProdSet Fp lin 2}
to derive the following result:

\begin{theorem}
\label{thm:Small_e_s} For a prime $p$ and a positive integer
$e\mid p-1$ with $e\le p^{\delta}$ for some fixed $\delta >0$,
given an oracle $\Oes$ for some
unknown $s\in \F_p$ and $t\in \F_p$, there is a deterministic
algorithm to decide whether $s = t$ in time
$e^{c_0\delta}(\log p)^{O(1)}$,
where $c_0$ is some absolute constant.
\end{theorem}

For $e\mid p-1$ with $e\le (p-1)/2$ we define $N(e)$ as the
largest $H$ such that for some $x\in\F_p$ and $r\in\F_p^*$ we have
$x+1,\ldots,x+H\in r\cG_e$. We see from the proofs of Theorems
~\ref{thm:Medium_e_s} and~\ref{thm:Large_e_s} that
\begin{equation}
\label{eq:Nlarge e}
 N(e)\le e^{1/4+o(1)}
\end{equation}
as $e\to\infty$.

Lemma~\ref{lem:ProdSet Fp lin 2} gives the
following improvement of~\eqref{eq:Nlarge e} for small $e$.
If $e\le p^{\delta}$ then
\begin{equation}
\label{eq:Nsmall e}
N(e) = O\(e^{c_0\delta}\).
\end{equation}
In particular,
$$ N(e)=e^{o(1)}$$
as $e = p^{o(1)}$ and $e\to\infty$.

\subsection{Shifted Power Identity Testing with Unknown $t$}

For large values of $e$ we have the following simple result.

\begin{theorem}
\label{thm:Large_e_st} For a prime $p$ and a positive integer
$e\mid p-1$ with $e\le (p-1)/2$,
given two oracles $\Oes$  and $\Oet$ for some unknown
$s,t\in \F_p$, there is a deterministic algorithm to
decide whether $s = t$ in time $p^{1/2+o(1)}$.
\end{theorem}

\begin{proof}
We note that by Lemma~\ref{lem:Weil2}, if $s\ne t$ then for $h =
\rf{p^{1/2}(\log p)^2}$ and sufficiently large $p$, the
condition~\eqref{eq:Cond Prelim} fails for at least one $x =1,
\ldots,h$. The algorithm is now immediate.
\end{proof}

For $e\le p^{3/4}$ we have a stronger result.

\begin{theorem}
\label{thm:Medium_e_st} For a prime $p$ and a positive integer
$e\mid p-1$ with $e\le (p-1)/2$,
given two oracles $\Oes$  and $\Oet$ for some unknown $s,t\in \F_p$,
there is a deterministic algorithm to
decide whether $s = t$ in time $ \max\{e^{1/2}p^{o(1)}, e^{2}p^{-1+o(1)}\}$.
\end{theorem}

\begin{proof} We fix some integer $h$ and
assume that~\eqref{eq:Cond Prelim} holds for every
$x \in \{0,\ldots,h\}$ and $s\ne t$.

Then there are
 $(h+1)^2$ values of   $x,y  \in \{0,\ldots,h\}$ we have
$$
\frac{(x+s)(y+s)}{(x+t)(y+t)} \in \cG_e.
$$
On the other hand, as we have shown in the proof of Theorem~\ref{thm:Small_e Oracle}
(see the bound~\eqref{eq: Qst}),
there are at most $e(h^{3/2+o(1)} p^{-1/2} + h^{o(1)})$ such pairs $(x,y)$.

Thus, fixing an arbitrary $\varepsilon > 0$ and
taking
$$
h = \max\{e^{1/2}p^{\varepsilon}, e^{2}p^{-1+\varepsilon}\},
$$
we see that~\eqref{eq:Cond Prelim} cannot hold  for every
$x \in  \{0,\ldots,h\}$ unless $s = t$.
Since $\varepsilon > 0$ is arbitrary, the result follows.
\end{proof}

Combining Theorems~\ref{thm:Large_e_st}  and~\ref{thm:Medium_e_st},
we obtain an algorithm of complexity
$$
%% \widetilde{T}
T_p(e)  =p^{o(1)}\left\{
\begin{array}{ll}
e^{1/2} & \text{ if } e< p^{2/3},\\
e^{2}p^{-1}& \text{ if }   p^{2/3}\le e < p^{3/4},\\
p^{1/2}& \text{ if }   p^{3/4}\le e\le (p-1)/2.
\end{array}
\right.
$$
In particular, $T_p(e) \le e^{2/3} p^{o(1)}$
for any $e\le (p-1)/2$.

For small values of $e$ we can use Lemma~\ref{lem:ProdSet Fp birat 2}
to derive the following result:

\begin{theorem}
\label{thm:Small_e_st} For a prime $p$ and a positive integer
$e\mid p-1$ with $e\le p^{\delta}$ for some fixed $\delta >0$,
given two oracles $\Oes$  and $\Oet$ for some unknown $s,t\in
\F_p$, there is a deterministic algorithm to decide whether $s =
t$ in time $e^{c_0\delta^{1/3}}(\log p)^{O(1)}$, where $c_0$ is
some absolute constant.
\end{theorem}

In particular, we see from Theorem~\ref{thm:Small_e_st} that if $e
= p^{o(1)}$ and $e\to\infty$ then we can test whether $s=t$ in
time $e^{o(1)} (\log p)^{O(1)}$ in $e^{o(1)}$ oracle calls.

\section{Comments and Open Questions}

Probably the most challenging question is to design
a deterministic algorithm for the Hidden Shifted Power Problem
which is faster than interpolation.

We note that the  constants
in Lemmas~\ref{lem:ProdSet Fp lin 2} and~\ref{lem:ProdSet Fp birat 2}
and the bound~\eqref{eq:Nsmall e},
can easily be made explicit. It is a natural question to obtain
good numerical values for these constants and thus fully explicit versions
of Theorem~\ref{thm:Small_e_s} and~\ref{thm:Medium_e_st}.

As we have mentioned, Lemma~\ref{lem:GProdSet Fp lin 1} solves an open problem from~\cite{CillGar}. Furthermore, the arguments used in the proof of Lemmas~\ref{lem:GProdSet Fp lin 1}
and~\ref{lem:ProdSet Fp birat 1} can be used for several other problems.
They can also be used to generalise and improve some of the results
of~\cite{BKS2} about intersections of intervals and subgroups
of $\F_p^*$.

We have proven that for any $e\le(p-1)/2$ the minimal number of calls
to oracle $\Oes$ to find $s$ is $p^{o(1)}$.
However, one can study a more
general problem. Let $\cA\subseteq\F_p$. We define  $\OAs$ as an
oracle that on every input $x \in \F_p$ outputs $1$ is $x+s\in \cA$
and $0$ otherwise, where $s$ is a ``hidden'' element $s \in \F_p$.

\begin{question}
\label{Arbset} Is it true that for any fixed $\delta\in(0,1/2)$
and for any set $\cA\subseteq\F_p$ with $\delta p\le \# A\le(1-\delta)p$ there is a deterministic algorithm that finds $s$ after
$p^{o(1)}$ calls (or even $O(\log p)$ calls) to the oracle $\OAs$
as $p\to\infty$?
\end{question}

It is also important for applications to pairing based
cryptography to extend our results to arbitrary finite fields. We
note that analogues of some of the results we have used are also
known for arbitrary finite fields. For example,  versions of
Lemma~\ref{lem:PVB} has recently been obtained for arbitrary
finite fields, see~\cite{Chang1,Chang2,Kon}.  Lemmas~\ref{lem:Weil1},
\ref{lem:Weil2}
and~\ref{lem:Weil3}
can also be easily extended to arbitrary fields.  However
analogues of many other results, such as
Lemmas~\ref{lem:CillGar},~\ref{lem:ShkVyu} and~\ref{lem:TripleSums} are not
known for arbitrary finite fields.

As we have mentioned, there are efficient quantum algorithms to solve
the Hidden Shifted Power Problem. However they require a {\it quantum\/}
oracle $\Oes$. It is certainly natural to investigate how much speed-up
quantum algorithms can provide in the case of a {\it classically\/} given
oracle $\Oes$ (that is, as in all results of this work).

Finally, it is also interesting to consider similar problems in
the case when the ``noisy'' oracles, which, with a certain
probability, for a given input
does not return any answer or even may return a wrong answer.

\section*{Acknowledgement}

The authors would like to thank Alfred Menezes for drawing out attention
to this problem, to Luis Pardo for discussions of the arithmetic
Nullstellensatz and to Lajos R{\'o}nyai for fruitful
discussions of the algorithm of Lemma~\ref{lem:BinEq}.

The idea of this work  started while the third author was visiting the
University of Waterloo whose
hospitality and perfect working conditions are gratefully appreciated.

The research of  J.~B. was partially supported
by National Science Foundation   Grant DMS-0808042, that of  S.~K.
by Russian Fund for Basic Research Grant N.~11-01-00329
and that of I.~S. by Australian Research Council Grant  DP1092835.

\end{document}